\documentclass[lettersize,journal,]{IEEEtran}
\ifCLASSOPTIONcompsoc
  \usepackage[nocompress]{cite}
\else
  \usepackage{cite}
\fi
\ifCLASSINFOpdf
\else
\fi

\usepackage{amsmath,amsfonts,amssymb,amsthm}
\usepackage{algorithmic}
\usepackage{array}
\usepackage{subcaption}
\usepackage{graphicx}
\usepackage{multirow}
\usepackage{textcomp}
\usepackage{stfloats}
\usepackage{url}
\usepackage{verbatim}
\usepackage{graphicx}
\usepackage{color, soul}
\usepackage{bm}
\usepackage{hyperref}

\usepackage{cite}
\usepackage{xcolor}
\usepackage{xpatch}

\newtheorem{theorem}{Theorem}

\usepackage{algorithmic}
\usepackage[linesnumbered,ruled]{algorithm2e}

\begin{document}
\title{User-Centric Communication Service Provision for Edge-Assisted Mobile Augmented Reality}

\author{    Conghao~Zhou,~\IEEEmembership{Member,~IEEE,}
            Jie~Gao,~\IEEEmembership{Senior~Member,~IEEE,}
            Shisheng~Hu,~\IEEEmembership{Member,~IEEE,}
            Nan~Cheng,~\IEEEmembership{Senior Member,~IEEE,}            
            Weihua~Zhuang,~\IEEEmembership{Fellow,~IEEE,}
            and Xuemin~(Sherman)~Shen,~\IEEEmembership{Fellow,~IEEE}

            \thanks{C.~Zhou is with the School of Telecommunications Engineering, Xidian University, China. He was previously with the Department of Electrical and Computer Engineering, University of Waterloo, Waterloo, ON N2L 3G1, Canada (e-mail:~c89zhou@uwaterloo.ca).}
            \thanks{J.~Gao is with the School of Information Technology, Carleton University, Ottawa, ON  K1S 5B6, Canada (email:~jie.gao6@carleton.ca).}
            \thanks{S.~Hu, W.~Zhuang, and X.~Shen are with the Department of Electrical and Computer Engineering, University of Waterloo, Waterloo, ON N2L 3G1, Canada (e-mail:\{s97hu, sshen, wzhuang\}@uwaterloo.ca).}
            \thanks{N. Cheng is with the School of Telecommunications Engineering, Xidian University, China (e-mail:~nancheng@xidian.edu.cn).}
            \thanks{Shisheng Hu is the corresponding author.}
            \thanks{Part of this work was presented at IEEE/CIC ICCC 2024~\cite{zhou2024user}.}
        }

\IEEEtitleabstractindextext{%
\begin{abstract}
Future 6G networks are envisioned to facilitate edge-assisted mobile augmented reality (MAR) via strengthening the collaboration between MAR devices and edge servers. In order to provide immersive user experiences, MAR devices must timely upload camera frames to an edge server for simultaneous localization and mapping (SLAM)-based device pose tracking. In this paper, to cope with user-specific and non-stationary uplink data traffic, we develop a digital twin (DT)-based approach for user-centric communication service provision for MAR. Specifically, to establish DTs for individual MAR devices, we first construct a data model customized for MAR that captures the intricate impact of the SLAM-based frame uploading mechanism on the user-specific data traffic pattern. We then define two DT operation functions that cooperatively enable adaptive switching between different data-driven models for capturing non-stationary data traffic. Leveraging the user-oriented data management introduced by DTs, we propose an algorithm for network resource management that ensures the timeliness of frame uploading and the robustness against inherent inaccuracies in data traffic modeling for individual MAR devices. Trace-driven simulation results demonstrate that the user-centric communication service provision achieves a 14.2\% increase in meeting the camera frame uploading delay requirement in comparison with the slicing-based communication service provision widely used for 5G.
\end{abstract}

}

\maketitle

\IEEEdisplaynontitleabstractindextext

\IEEEpeerreviewmaketitle

\section{Introduction}\label{sec1}

Immersive communications are expected to extend the enhanced Mobile Broadband (eMBB) communications of 5G networks for providing users with rich immersive experiences in future 6G networks~\cite{zhou2024user_wcm}. Augmented reality (AR), a typical form of immersive communications, enables seamless integration of virtual objects into the physical environment surrounding human users. Driven by the increasing demand for immersive experiences, AR accessible on portable devices such as smart glasses, referred to as mobile AR (MAR), is gaining widespread attention as one of the emerging applications in the 6G era. All MAR applications need device pose tracking, a resource-intensive procedure fundamental to the 3D alignment of virtual objects with the physical environment~\cite{jin2023ebublio,huzaifa2021illixr}. Enabling device pose tracking is a key challenge for current MAR devices due to their resource limitations such as limited battery power. To overcome this challenge, leveraging the resources of edge servers via edge-assisted MAR is a promising paradigm~\cite{chen2023networked}.

Although communication service provision for edge-assisted applications has been investigated in the context of 5G networks~\cite{wang2022leaf,wang2017joint}, supporting edge-assisted device pose tracking for MAR applications in future 6G networks requires novel solutions for the following two reasons~\cite{shen2021holistic}. First, human movements such as changing gaze points are incessant in MAR applications, and differences in the movements of different users result in distinctive network resource demands even if they are using the same MAR application~\cite{ran2020multi}. Most existing communication service provision approaches rely on \emph{service-oriented} demand modeling, which cannot capture the variances in service demands across different users to support personalized MAR user experiences. For example, slicing-based communication service provision widely adopted in 5G networks for aggregated video-type data traffic transmission~\cite{navarro2020survey} lacks the capability to characterize the video traffic patterns of individual users. Second, to deal with the uncertainties in human movements, current MAR has incorporated a complex localization-based operational mechanism~\cite{chen2023networked}, i.e., simultaneous localization and mapping (SLAM)-based device pose tracking~\cite{campos2021orb} with multiple interacting functionality modules, from the perspective of the application layer. Current communication service provision approaches developed from the networking perspective mostly overlook the impact of this intricate operational mechanism for MAR applications~\cite{li2022risk,li2020data}, thereby unable to natively cope with uncertainties in human movements~\cite{wu2023characterizing,shen2023toward}. As a result, \emph{user-centric} communication service provision that can properly take the impact of MAR operational mechanism into account and characterize user-specific service demands is essential for future 6G networks to enhance edge-assisted device pose tracking in MAR.

The digital twin (DT) technique can support user-centric communication service provision for device pose tracking in MAR via user device virtualization~\cite{shen2021holistic}. The use of DT in communications and networking is still in its infancy. We explore establishing DTs as digital representations of individual MAR devices. By doing so, we leverage DTs to obtain and process high-quantity, high-quality data related to MAR devices~\cite{zhou2024user_wcm}. Unlike mathematical methods that explicitly model certain characteristics, such as the statistical frame uploading pattern of an MAR device, DTs can capture implicit yet crucial information about the characteristics using data-driven techniques~\cite{zhou2024digital}. Additionally, DTs can be customized to provide the flexibility required to capture the unique characteristics of individual entities, e.g., the service demand pattern of each MAR device~\cite{sun2024knowledge}.

Despite these advantages, research on DTs in communication and networking is still in its infancy. Three issues arise when establishing a user DT supporting user-centric communication service provision for edge-assisted MAR. First, establishing a proper data model that effectively captures the impact of various factors in the intricate SLAM-based device pose tracking on the user-specific service demands is challenging~\cite{linowes2017augmented}. The selection of key factors for proper demand modeling and the organization of the data corresponding to these factors in a structured way are not explored yet for MAR applications. Second, temporal variations in user movements may lead to non-stationary service demand in MAR. For example, the data traffic load for uploading camera frames may surge intermittently following an event of device pose tracking loss~\cite{ran2020multi}. Non-stationary data traffic presents a challenge to effectively characterizing user-specific service demands. Third, since no modeling technique can perfectly capture user-specific service demands, DT-based communication service provision decisions should account for modeling inaccuracies~\cite{navarro2020survey,hu2023adaptive}. However, obtaining a proper representation of modeling inaccuracies, which is crucial for optimizing DT-based communication service provision, is also a challenge.

In this paper, we investigate communication service provision to facilitate timely camera frame uploading, which is essential for supporting SLAM-based device pose tracking in edge-assisted MAR. To address the aforementioned challenges, we establish an aug\underline{\textbf{M}}ented re\underline{\textbf{A}}lity use\underline{\textbf{R}} digita\underline{\textbf{L}} tw\underline{\textbf{IN}} (MARLIN), building on our general DT framework~\cite{shen2021holistic}, and propose a DT-based user-centric communication service provision approach to support edge-assisted device pose tracking for MAR in 6G networks. MARLIN consists of a customized data model to characterize the service demand of each individual MAR device and various DT operation functions to update the data model. The contributions of this paper are summarized as follows:
    \begin{itemize}
        \item We establish a data model customized for edge-assisted MAR, including well-structured data attributes and differentiated data update mechanisms, which supports capturing the unique service demands of individual MAR devices flexibly.

        \item We define two novel DT operation functions, which jointly capture the implicit impact of the operational mechanism of SLAM-based device pose tracking on the non-stationary uplink data traffic from each MAR device. One of them implements two data-driven traffic modeling methods with different complexities, and the other adaptively switches between these two methods in response to data traffic variations. 

        \item We propose a user-centric communication service provision algorithm based on MARLIN, which ensures robustness against the inherent inaccuracies resulting from our user-level data traffic modeling.

        \item We demonstrate the potential of MARLIN in facilitating user-centric communication service provision, by flexibly adapting data granularity to various data-driven modeling methods, through trace-driven simulations.

    \end{itemize}

The remainder of this paper is organized as follows. Section~\ref{sec2} summarizes related works. Section~\ref{sec3} presents the considered scenario, system models, and problem formulation. Section~\ref{sec4} and Section~\ref{sec5} introduce MARLIN and the proposed DT-based user-centric communication service provision, respectively. Section~\ref{sec6} presents the simulation results, and Section~\ref{sec7} concludes this paper.

\section{Related Works}\label{sec2}

In this section, we first review the related literature on device pose tracking techniques used in MAR. Then, we summarize existing research works on communication service provision from the networking perspective. 

\subsection{Device Pose Tracking in MAR}

To align virtual objects with the physical environment within the field of view of each MAR device, tracking the MAR device's pose, including its position and orientation, is fundamental~\cite{han2022comic}. The SLAM-based device pose tracking techniques and their variants are widely applied in MAR~\cite{linowes2017augmented}. Such techniques enable accurate 3D alignment by leveraging the 3D location information of physical objects. A trending research direction is using various visual features, e.g., binary robust independent elementary features (BRIEF)~\cite{campos2021orb} in the design of SLAM or MAR systems to enhance the tracking accuracy\cite{khosoussi2019reliable,davison2007monoslam}. Many researchers focus on combining traditional visual features and information collected from additional sensors, e.g.,~LiDAR~\cite{chou2021efficient} and inertial measurement units~\cite{carlone2018attention}, to improve the tracking accuracy. As a result, the enhanced tracking accuracy comes at the cost of increased operational complexity and resource demands, thereby posing a significant challenge to resource-constrained mobile devices~\cite{zhou2024digital}. 

To overcome the resource constraints of MAR devices, many recent studies have adopted cloud or mobile edge computing paradigms~\cite{chen2018marvel}. Specifically, the SLAM system for device pose tracking can be divided into two modules: a computing-intensive mapping module and a lightweight tracking module, which are handled by a cloud/edge server and a mobile device, respectively~\cite{chen2023networked}. Ben Ali~\emph{et~al.} develop a practical edge-assisted SLAM system for the mobile device, demonstrating the feasibility and potential of edge-assisted device pose tracking~\cite{ben2022edge}. The authors in~\cite{chen2023adaptslam} and~\cite{zhou2024digital} design adaptive edge-assisted device pose tracking schemes aim to overcome constrained uplink data rates and time-varying uplink data rates, respectively, between the edge server and an MAR device. In addition, to support accurate virtual object rendering in multi-user MAR applications, existing studies also investigate coordinate synchronization for the device pose tracking of different MAR devices~\cite{ran2020multi,dhakal2022slam,ren2020edge}. 

Different from aforementioned studies that advance device pose tracking from the application perspective, our work emphasizes service provision from the communication and networking perspective to support camera frame uploading, which is crucial for ensuring the timeliness of device pose tracking in edge-assisted MAR.

\subsection{Communication Service Provision}

Communication service provision in existing works concentrates on configuring networks to manage network resources for meeting service demands. Although communication service provision is a classic research topic in networking, its significance persists due to the increasingly diversified and demanding services~\cite{shen2021holistic}.

In the early phase of wireless communication network research, communication service provision primarily focuses on the on-demand allocation of communication resources such as the radio spectrum and power at each base station. By modeling communication service demands, e.g.,~data packet arrival patterns and device location distributions, communication resources can be allocated according to the service demand modeling while considering constraints on resources~\cite{lin2018towards}, interference management~\cite{wu2022demand}, etc. With the advent of 5G, mobile edge computing (MEC) results in a plethora of computing services that require communication, computing, and data storage resources. Some works investigate multi-resource management problems for MEC in network scenarios such as drone-assisted networks~\cite{zhang2020latency,qu2021service,zhang2021aoi} and terrestrial networks~\cite{qu2021resilient}. Meanwhile, network slicing, as an innovation in 5G, enables the establishment of isolated virtual networks, i.e., slices, over the same physical network infrastructure and differentiated resource management for different services to satisfy their respective service requirements~\cite{ye2018dynamic}. As a result, slicing-based communication service provision approaches are developed for various network scenarios~\cite{jia2021vnf,li2020data,qiu2022online,li2022risk,wang2020sfc,ye2018dynamic}. These works focus on minimizing the network resource consumption while satisfying service requirements, including the timely completion of dynamically requested computing tasks~\cite{jia2021vnf}, the reliability of resource management~\cite{qiu2022online,li2022risk}, and end-to-end delay requirements~\cite{wang2020sfc,ye2018dynamic}.

Despite these advancements, few studies have investigated communication service provision for MAR, especially for device pose tracking. Since the MAR operational mechanism has a significant impact on the service demands of MAR devices, existing communication service provision approaches that overlook the influence of the intrinsic MAR operational mechanism may fail to accurately capture MAR service demands~\cite{chen2023networked}. In addition, the slicing-based communication service provision primarily targets meeting the aggregated service demands of users running the same application, rather than user-specific service demands. As a result, they are not suitable for MAR applications, which involve uncertain user-specific movements~\cite{li2022risk,li2020data}.  

Different from existing works, we aim to achieve user-centric communication service provision for MAR applications. Specifically, we establish a DT for each individual MAR device, which supports a data-driven method to characterize the impact of the complex MAR operational mechanism on user-specific service demands.    

\section{Edge-assisted Device Pose Tracking in MAR}\label{sec3}

\subsection{Considered Scenario}

When a user runs an MAR application, the position and orientation (jointly referred to as the \emph{3D pose}) of the MAR device change over time due to the user movements. The MAR device captures camera frames periodically with a fixed frame rate and tracks its 3D pose based on the captured camera frames, which is crucial for rendering virtual objects at correct locations within the user’s field of view~\cite{campos2021orb}. 

    \begin{figure}[t]
        \centering
        \includegraphics[width=0.49\textwidth]{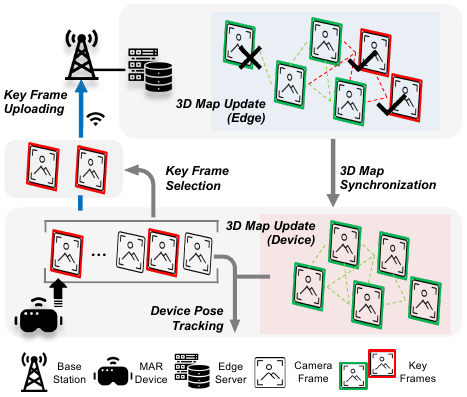}
        \caption{An illustration of edge-assisted device pose tracking for one MAR device.}\label{system}
    \end{figure}

An edge server deployed at a base station (BS) assists the device pose tracking of a set of MAR devices, denoted by~$\mathcal{N}$, within the communication coverage of the BS. Currently, SLAM techniques are widely adopted to support device pose tracking in MAR~\cite{chen2023networked}. An example paradigm of edge-assisted device pose tracking for one MAR device is illustrated in Fig.~\ref{system}, in which the edge server and the device collaboratively track its 3D pose. Specifically, the edge server includes a resource-intensive mapping module to construct and update a 3D representation (i.e., a 3D map) of the physical environment surrounding each MAR device for enabling device pose tracking, while the MAR device includes a lightweight tracking module for calculating its real-time 3D pose based on the 3D map~\cite{chen2023adaptslam}. For each MAR device, both the device and the edge server maintain their own 3D maps during device pose tracking. For clarity, we refer to them as the ``3D map (device)'' and ``3D map (edge)'' in this paper, respectively.

As shown in Fig.~\ref{system}, edge-assisted device pose tracking can be divided into four key steps~\cite{ben2022edge}: (i) each MAR device selects a subset of recently captured camera frames as the key frames and uploads these key frames to the edge server over a wireless communication link; (ii) the mapping module equipped at the edge server updates the 3D map (edge) by incorporating newly uploaded key frames along with previously stored ones; (iii) the edge server sends the updated 3D map (edge) to this MAR device for 3D map synchronization; and (iv) the tracking module at each MAR device updates the 3D map (device) based on the information on the received 3D map (edge) and leverages the updated 3D map (device) to locally calculate the device pose for every camera frame. These four steps are performed iteratively during device pose tracking, as detailed in the following subsections.

\subsection{Key Frame Uploading}\label{sec32}

To cope with uncertain user movements, each MAR device with a fixed frame rate periodically selects a number of key frames from recently captured camera frames and uploads them to the edge server for updating the 3D map (edge). Let $\mathcal{F}_{n}$ denote the set of all camera frames captured by MAR device~$n \in \mathcal{N}$. We refer to the duration of $F$ consecutive camera frames as a time slot, i.e., the period of key frame uploading, and denote the set of all time slots by~$\mathcal{T}$. Let~$\mathcal{F}_{n,t} \subseteq \mathcal{F}_{n}$ denote the set of camera frames captured during time slot~$t \in \mathcal{T}$. 

For camera frame~$f \in \mathcal{F}_{n}$ of MAR device~$n$, we use~$a_{n,f} \in \left\{0,1 \right\}$ to indicate whether it is a key frame or not. Let $a_{n, f} = 1$ if camera frame~$f$ is selected as a key frame, and $a_{n, f} = 0$ otherwise. The detail about key frame selection will be presented in Subsection~\ref{sec33}. Denote the set of key frames selected in time slot~$t$ by~$\mathcal{K}_{n,t} = \left\{f \in \mathcal{F}_{n, t} | a_{n,f} = 1 \right\}$. Considering the time-varying number of key frames for uploading~\cite{ben2022edge}, we define the number of key frames in each time slot as a random variable, given by: 
    \begin{equation}\label{eq1}
        \tilde{k}_{n,t} = \sum_{f \in \mathcal{F}_{n, t}}{a_{n,f}}, \,\, \forall t \in \mathcal{T}.
    \end{equation}
At the end of time slot~$t$, MAR device~$n$ uploads the set of selected key frames,~i.e.,~$\mathcal{K}_{n,t}$, to the edge server. 

Proactive radio spectrum resource management is necessary to ensure timely key frame uploading for facilitating real-time device pose tracking in MAR~\cite{chen2023networked}. Denote the radio spectrum bandwidth reserved for MAR device~$n$ in time slot~$t$ for key frame uploading by~$b_{n,t}$. For proactive radio spectrum resource management, the value of~$b_{n,t}$ should be determined at the beginning of each time slot. Given~$b_{n,t}$, the uplink data rate of MAR device~$n$ within time slot~$t$, denoted by~$r_{n,t}$, is as follows:
    \begin{equation}\label{}
        r_{n,t} = b_{n,t} \log (1 + \gamma_{n,t}), \,\, \forall t \in \mathcal{T}, n \in \mathcal{N},
    \end{equation}
where~$\gamma_{n,t}$ represents the estimated signal-to-noise ratio (SNR) at MAR device~$n$ in time slot~$t$~\cite{belgiovine2021deep}. We denote the volume of data (in bits) for uploading each camera frame by~$\alpha$, the same for all camera frames. 

To achieve user-centric communication service provision for MAR, we target ensuring the timely key frame uploading of each individual MAR device, rather than focusing on the performance averaged over multiple devices as in slicing-based communication service provision~\cite{li2020data}. Uploading the set of key frames selected by MAR device~$n$ in time slot~$t$, i.e.,~$\tilde{k}_{n,t}$, the following constraint should be satisfied~\cite{atawia2016joint}:
  
    \begin{equation}\label{eq2}
        P ( \alpha \tilde{k}_{n,t} \leq T^\text{r} r_{n,t} ) \ge \epsilon, \,\, \forall t \in \mathcal{T}, n \in \mathcal{N},
    \end{equation}
where~$\epsilon \in [0, 1]$ represents the required reliability of communication service provision, and~$T^\text{r}$ represents the maximum tolerable total transmission duration for uploading the selected key frames from an MAR device. 

\subsection{3D Map Update \& Synchronization}\label{sec33}

The key frames uploaded by each MAR device are used to form or update the 3D map (edge) at the edge server. The 3D map consists of a set of key frames uploaded by the MAR device over time as well as the feature points (FPs), e.g., a wall corner, detected from each key frame~\cite{campos2021orb}.\footnote{A 3D map consists of a set of key frames and the set of FPs identified in each key frame~\cite{campos2021orb}. Since the update of FPs is determined by the update of key frames, we use the adding and deleting of key frames to represent the update of a 3D map for simplicity.} The updates of the two 3D maps are different but correlated.

The set of key frames stored in the 3D map (edge) of MAR device~$n \in \mathcal{N}$ at time slot~$t$, denoted by~$\mathcal{M}^\text{E}_{n,t} \subseteq \mathcal{F}$, changes as shown in the following equation~\cite{chen2023adaptslam,zhou2024digital}:
    \begin{equation}\label{eq3}
        \mathcal{M}^\text{E}_{n,t} = \left\{ \mathcal{M}^\text{E}_{n,t-1} \cup \mathcal{K}_{n,t-1} \right\} \backslash \mathcal{C}_{n,t-1},
    \end{equation}
where~$\mathcal{C}_{n,t-1} \subseteq \mathcal{M}^\text{E}_{n,t-1}$ represents the set of key frames removed from the 3D map (edge)~$\mathcal{M}^\text{E}_{n,t-1}$ of MAR device~$n$ in time slot~$t$. We assume that the update of~$\mathcal{M}^\text{E}_{n,t}$ can be accomplished at the end of each time slot. After the 3D map (edge) update, the set~$\mathcal{M}^\text{E}_{n,t}$ is downloaded by MAR device~$n$ for updating the 3D map (device)~\cite{ben2022edge}.

In MAR applications, key frame selection and uploading is conducted based on an implicit policy, denoted by~$\Pi$, used by the MAR operational mechanism of device pose tracking~\cite{campos2021orb}. Generally, in a policy, a frame is selected as a key frame if it differs sufficiently from its preceding key frames and thus can provide adequate new environment information, e.g., FPs, while having sufficient overlap with preceding for easily implementing device pose tracking~\cite{campos2021orb}. While factors influencing key frame selection and uploading may differ in existing policies~\cite{ben2022edge,chen2023adaptslam,zhou2024digital}, a general policy~$\Pi$ can be formulated as follows:
    \begin{equation}\label{eq4}
        a_{n,f} = \Pi (\mathcal{M}^\text{D}_{n,f}, \mathcal{X}_{n,f}, f),\,\, \forall f \in \mathcal{F}_{t}, t \in \mathcal{T},
    \end{equation}
where~$\mathcal{M}^\text{D}_{n,f}$ denotes the set of key frames in the 3D map (device) of MAR device~$n$, and~$\mathcal{X}_{n,f} = \left\{ f' \in \mathcal{F} | f-X \leq f' < f  \right\}$ represents the set of $X$ preceding camera frames when camera frame~$f$ is captured by MAR device~$n$. Given the value of~$a_{n,f}$, the 3D map (device) is updated as follows:
    \begin{equation}\label{eq41}
         \mathcal{M}^\text{D}_{n,f} = 
        \begin{cases} 
            \mathcal{M}^\text{D}_{n,f-1} \cup \left\{  f \right\} & \text{if } a_{n,f-1} = 1; \\
            \mathcal{M}^\text{D}_{n,f-1} & \text{otherwise }.
        \end{cases}
    \end{equation}
Given \eqref{eq4} and~\eqref{eq41}, key frame selection and uploading, depending on policy~$\Pi$, can be abstracted as a sequential decision making process~\cite{zhou2024digital}.

\subsection{Problem Formulation}

To efficiently support edge-assisted device pose tracking in MAR, we formulate a user-centric communication service provision problem with the objective of minimizing the radio spectrum resource reserved for the key frame uploading of all MAR devices, as follows:
    \begin{subequations}\label{p1}
        \begin{align}
            \textrm{P1:} &\,\, \min_{ \{ b_{n,t} \}_{n \in \mathcal{N}, t \in \mathcal{T}} } \sum_{n \in \mathcal{N} }{\sum_{ t \in \mathcal{T} }{ b_{n,t} }}\\
            \textrm{s.t.} &\,\, P( \alpha \tilde{k}_{n,t} \leq T^\text{r} r_{n,t} ) \ge \epsilon, \,\, \forall n \in \mathcal{N}, t \in \mathcal{T},
        \end{align}
    \end{subequations}
where constraint~(\ref{p1}b) ensures that the key frame uploading of each individual MAR device satisfies the delay requirement with a sufficiently high probability.\footnote{We focus on the camera frame uploading delay that is affected by decisions for the data plane, e.g., network resource provisioning decisions, rather than delay stemming from the control plane, e.g.,~control signaling delay.} In Problem~P1, random variables~$\tilde{k}_{n,t}$ are unknown~\emph{a priori}, and temporal variations in data traffic of each MAR device may be non-stationary. Existing communication service provision approaches optimize radio spectrum resource reservation based on either mathematical modeling or data-driven prediction~\cite{navarro2020survey}. However, mathematically modeling~$\tilde{k}_{n,t}$ is intractable due to the implicit policy~$\Pi$. While data-driven approaches, e.g.,~deep neural networks (DNNs), potentially outperform mathematical modeling in the accuracy of approximating a complex policy~$\Pi$, existing approaches overlook the impact of the MAR operational mechanism of device pose tracking on the data traffic for key frame uploading~\cite{chen2023adaptslam,ben2022edge}, and thereby struggle to adapt to non-stationary data traffic. 

To solve Problem~P1, we develop a novel digital twin (DT)-based approach, as shown in Fig.~\ref{udt}. First, we establish a user DT that can characterize the impact of the MAR operational mechanism on the data traffic from individual MAR devices. Second, we propose a user-centric communication service provision method with the consideration of potential modeling inaccuracies. The two modules are presented in Sections~\ref{sec4} and~\ref{sec5}, respectively.

\section{MAR User Digital Twin Establishment}\label{sec4}

In this section, we customize an aug\underline{\textbf{M}}ented re\underline{\textbf{A}}lity use\underline{\textbf{R}} digita\underline{\textbf{L}} tw\underline{\textbf{IN}}, called MARLIN, for each individual MAR device. Our designs for MARLIN evolves from the framework presented in~\cite{zhou2024digital,zhou2024user_wcm}. As illustrated by the light grey block in Fig.~\ref{udt}, MARLIN, comprising an \emph{MAR user profile} (MUP) and two \emph{user DT operation functions} (UDTOFs), is deployed at the BS and maintained by the controller to facilitate communication service provision for MAR. In the MUP, we define a \emph{data model} customized for MAR, consisting of a set of well-structured data attributes chosen to characterize the service demand of the corresponding MAR device and to support our user-centric communication service provision. The two UDTOFs are leveraged to process the collected raw data in various ways for updating the data within the MUP accordingly. We introduce the two UDTOFs first, followed by the MUP.

    \begin{figure}[t]
        \centering
        \includegraphics[width=0.49\textwidth]{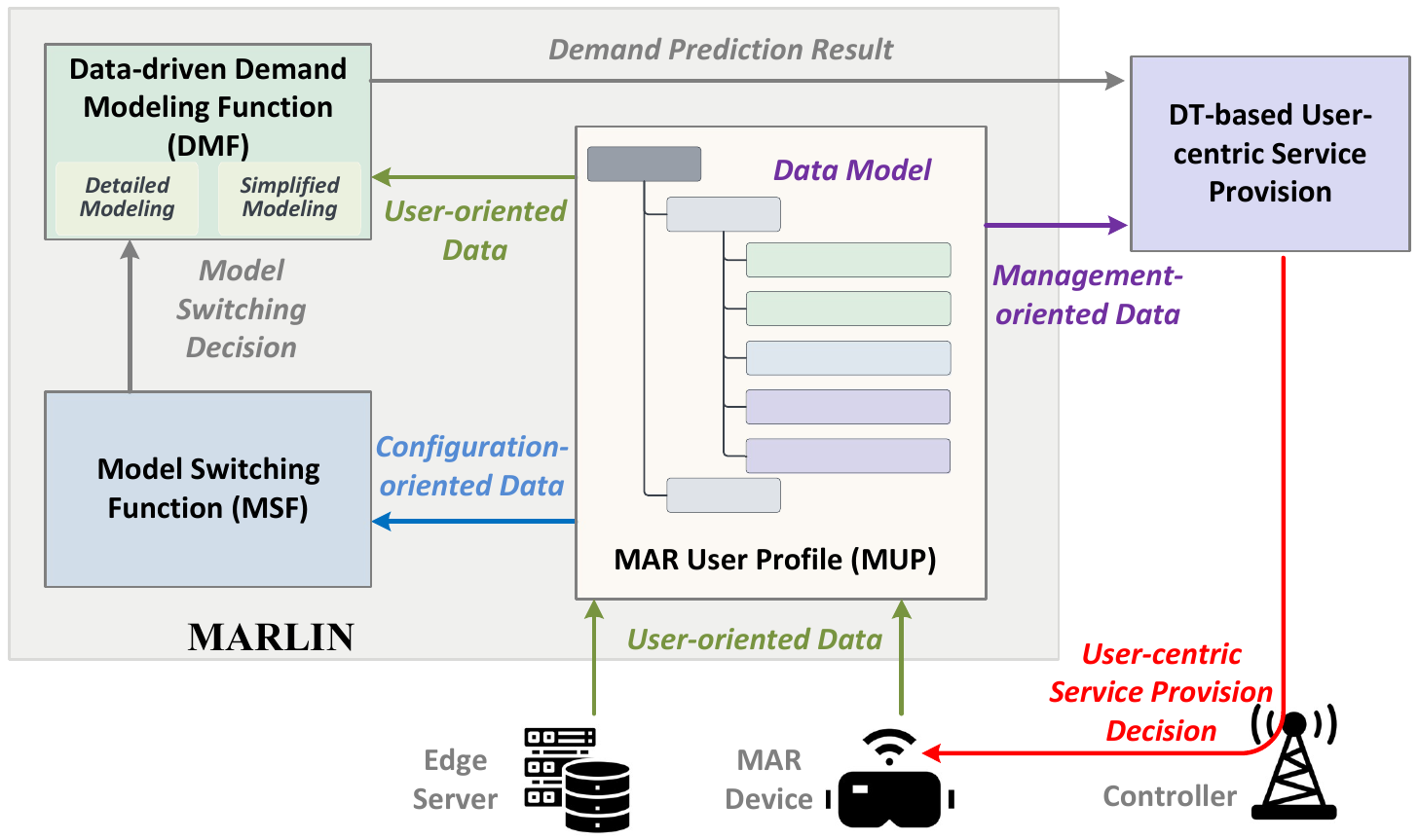}
        \caption{The workflow of the proposed DT-based user-centric communication service provision approach.}\label{udt}
    \end{figure} 

\subsection{Data-driven Demand Modeling Function (DMF)}

User-centric communication service provision for MAR requires an accurate model to capture the uplink data traffic of key frame uploading at each individual MAR device. To this end, we use a Markov decision process to model the sequential decision making of each MAR device on key frame selection and uploading. 

For MAR device~$n$, we denote a state by~$\mathbf{s}_{n,f} \in \mathcal{S}_{n}$ and define~$a_{n,f} \in $ as an action, i.e., determining whether camera frame~$f$ is a key frame or not, based on the state. Define state transition probability function~$P(\mathbf{s}_{n,f+1}|\mathbf{s}_{n,f}, a_{n,f}) := \mathcal{S}_{n} \times \mathcal{A}_{n} \rightarrow \mathcal{S}_{n}$. The choice of the action~$a_{n,f}$ depends on the policy~$\Pi$ mentioned in~\eqref{eq4}. Given actions~$\mathbf{a}_{n,t} = [a_{n, f}]_{\forall f \in \mathcal{F}_{n,t}} \in \mathcal{A}_{n}$ in time slot~$t$, the data traffic of MAR device~$n$ for key frame uploading in time slot~$t$ can be obtained from~\eqref{eq1}. 

The goal of our demand modeling function (DMF) is to approximate policy~$\Pi$ for predicting the data traffic of key frame uploading. A foundation for approximating~$\Pi$ is to define the states. Considering model accuracy and model complexity, we introduce two types of states for a \emph{detailed} and a \emph{simplified} data traffic modeling in the DMF, respectively. In addition to the approximation of policy~$\Pi$, the DMF should approximate the state transition probability function~$P(\mathbf{s}_{n,f+1}|\mathbf{s}_{n,f}, a_{n,f})$ to support data traffic modeling.

\subsubsection{Detailed Modeling}

In SLAM-based device pose tracking, the relations among the sets of FPs corresponding to camera frames are essential for estimating the device poses at the instants of capturing these camera frames, thereby significantly affecting key frame selection and uploading in MAR~\cite{khosoussi2019reliable}. 

In our detailed modeling, we use a graph to model the relations among the sets of FPs corresponding to camera frames and capture the impact of such a graph on the data traffic load of each MAR device~\cite{chen2023adaptslam}. 

Denote the set of FPs identified in camera frame~$f$ by~$\mathcal{U}_{f}$. We model the 3D map (device) of MAR device~$n$ and the camera frame~$f$, i.e.~$\mathcal{M}^\text{D}_{n, f} \cup \left\{ f \right\}$, as a weighted undirected graph, denoted by~$\mathcal{G}^\text{D}_{n,f} = (\mathcal{M}^\text{D}_{n, f} \cup \left\{ f \right\}, \mathcal{E}^\text{D}_{n,f})$, where~$\mathcal{E}^\text{D}_{n,f}$ represents the set of edges between every pair of camera frames in~$\mathcal{M}^\text{D}_{n, f} \cup \left\{ f \right\}$. For edge~$e = (f, f') \in \mathcal{E}^\text{D}_{n,f}$ connecting camera frames~$f, f' \in \mathcal{M}^\text{D}_{n,f}$, its weight is defined as the Jaccard coefficient, denoted by~$\varepsilon_{f, f'} \in [0,1]$, which quantifies the similarity of sets~$\mathcal{U}_{f}$ and~$\mathcal{U}_{f'}$ as follows:
    \begin{equation}\label{}
        \varepsilon_{f, f'} = \frac{|\mathcal{U}_{f} \cap \mathcal{U}_{f^{'}}|}{|\mathcal{U}_{f} \cup \mathcal{U}_{f^{'}}|}, \,\,\,\, \forall\,\, \mathcal{U}_{f} \cup \mathcal{U}_{f^{'}} \neq \emptyset,
    \end{equation} 
where~$|\cdot|$ represents the cardinality of a set. When many observations are shared by camera frames~$f$ and $f'$, the two sets of FPs~$\mathcal{U}_{f}$ and $\mathcal{U}_{f'}$ are similar, resulting in a large~$\varepsilon_{f, f'}$. Similarly, we model the set of camera frames~$\mathcal{X}_{n,f}$ defined right after~\eqref{eq4} as another weighted undirected graph, denoted by~$\mathcal{G}_{n,f} = (\mathcal{X}_{n,f}, \mathcal{E}_{n,f})$. Based on these graphs, we define the state used in the detailed modeling for MAR device~$n$ as follows:
    \begin{equation}\label{}
        \mathbf{s}^\text{D}_{n,f} = [\mathcal{G}^\text{D}_{n,f}, \mathcal{G}_{n,f}], \,\, \forall n \in \mathcal{N}, f \in \mathcal{F}_{n},
    \end{equation} 
which is abbreviated as the \emph{D-state}.

We adopt a graph convolutional network (GCN)-based technique to approximate policy~$\Pi$ by using the D-state~$\boldsymbol{\vartheta}^\text{D}$ as the input and action~$a_{n,f}$ as the output. Let~$\boldsymbol{\vartheta}^\text{D}$ denote the parameters of the GCN, and define~$\Pi^\text{D}$ as the policy generating action by using this GCN. To make policy~$\Pi^\text{D}$ approach policy~$\Pi$, we optimize the parameters~$\boldsymbol{\vartheta}^\text{D}$ by minimizing the following loss function:
    \begin{equation}\label{}
        L(\boldsymbol{\vartheta}^\text{D}) = \frac{1}{|\Xi|} \sum_{(a_{n,f},\mathbf{s}^\text{D}_{n,f}) \in \Xi} \left(a_{n,f} - \Pi^\text{D}(\mathbf{s}^\text{D}_{n,f};{\boldsymbol{\vartheta}^\text{D} }) \right)^{2},
    \end{equation}
where~$\Xi$ represents a set containing historical information on actions and D-states, stored in the MUP for MAR device~$n$. 

For graphs~$\mathcal{G}^\text{D}_{n,f}$ and~$\mathcal{G}_{n,f}$, we use two GCNs for feature extraction. The extracted features are concatenated, followed by fully-connected layers. The computational complexity for predicting whether a camera frame will be uploaded using the D-model is given by~\cite{kipf2016semi,10771784}:
    \begin{equation}\label{}
        O(Z^\text{D}) = O(|\mathcal{E}_{n,f}|N^\text{D}_{\text{I},1} N^\text{D}_{\text{H},1} N^\text{D}_{\text{O},1} + |\mathcal{E}^\text{D}_{n,f}| N^\text{D}_{\text{I},2} N^\text{D}_{\text{H},2} N^\text{D}_{\text{O},2} + N^\text{D}_{\text{FL}}), 
    \end{equation} 
where $N^\text{D}_{\text{I},1}$ and $N^\text{D}_{\text{I},2}$ denote the input feature dimensions of the nodes, $N^\text{D}_{\text{H},1}$ and $N^\text{D}_{\text{H},2}$ denote the hidden layer dimensions of the two GCNs, $N^\text{D}_{\text{O},1}$ and $N^\text{D}_{\text{O},2}$ denote the dimensions of features extracted by the GCNs, and~$N^\text{D}_{\text{FL}}$ represents the number of parameters in the fully-connected layers. Given~$F$ camera frames to be uploaded within the subsequent time slot, the computational complexity becomes~$O(Z^\text{D} F)$.

\subsubsection{D-State Transition Modeling}

In the detailed modeling, the state transition probability in D-state~$\mathbf{s}^\text{D}_{n,f}$ is characterized as follows:
    \begin{subequations}\label{eq8}
        \begin{align}
            & P(\mathbf{s}^\text{D}_{n,f+1}| \mathbf{s}^\text{D}_{n,f}, a_{n,f}) =  \frac{P(\mathcal{G}^\text{D}_{n,f+1}, \mathcal{G}_{n,f+1}, \mathcal{G}^\text{D}_{n,f}, \mathcal{G}_{n,f}, a_{n,f})}{P(\mathcal{G}^\text{D}_{n,f}, \mathcal{G}_{n,f}, a_{n,f})}\\
            & = \frac{P(\mathcal{G}^\text{D}_{n,f+1}, \mathcal{G}^\text{D}_{n,f}, a_{n,f}) P(\mathcal{G}_{n,f+1},\mathcal{G}_{n,f}|\mathcal{G}^\text{D}_{n,f+1}, \mathcal{G}^\text{D}_{n,f}, a_{n,f})}{P(\mathcal{G}^\text{D}_{n,f+1}, a_{n,f}) P(\mathcal{G}_{n,f} | \mathcal{G}^\text{D}_{n,f+1}, a_{n,f})}\\
            & = P(\mathcal{G}^\text{D}_{n,f+1}| \mathcal{G}^\text{D}_{n,f}, a_{n,f}) \frac{P(\mathcal{G}_{n,f+1},\mathcal{G}_{n,f})}{P(\mathcal{G}_{n,f})}\\
            & = P(\mathcal{G}^\text{D}_{n,f+1}| \mathcal{G}^\text{D}_{n,f}, a_{n,f}) P(\mathcal{G}_{n,f+1}| \mathcal{G}_{n,f}),
        \end{align}
    \end{subequations}
where~(\ref{eq8}c) holds due to the fact that newly arrived camera frames in~$\mathcal{X}_{n,f+1}$ in graph~$\mathcal{G}_{n,f+1}$ do not depend on graph of the 3D map (device), i.e.,~$\mathcal{G}^\text{D}_{n,f}$, or action~$a_{n,f}$. Since~$P(\mathcal{G}^\text{D}_{n,f+1}| \mathcal{G}^\text{D}_{n,f}, a_{n,f})$ can be calculated according to \eqref{eq41}, we need to approximate only~$P(\mathcal{G}_{n,f+1}| \mathcal{G}_{n,f})$ for modeling D-state transitions. We build a GCN ~$\phi(\mathcal{G}_{n,f};\boldsymbol{\theta})$ with parameters~$\boldsymbol{\theta}$ to approximate the D-state transitions. Note that this GCN is used only for graph link prediction, i.e.,~output only the weights of edges between camera frames, instead of the whole camera frames~\cite{zhang2019graph}.  

\subsubsection{Simplified Modeling}

The detailed modeling involves shared observations, e.g., the shared set of FPs, among historical camera frames. Excessive input data may introduce redundancy and decrease service demand modeling accuracy. For example, the procedure of key frame selection and uploading in the MAR operational mechanism is simple for device pose tracking when the variation in device pose is insignificant~\cite{campos2021orb,ben2022edge}. The accuracy of detailed modeling, thus, may be degraded since such redundant input data may complicate the GCN training for approximating policy~$\Pi$ in the detailed modeling.

As an alternative, we propose a simplified data-driven data traffic modeling and define the corresponding state used in the simplified modeling for MAR device~$n$, abbreviated as \emph{S-state}, as follows:
    \begin{equation}\label{}
        \mathbf{s}^\text{S}_{n,f} = [a_{n,i}]_{\forall f-T^\text{w} \leq i < f}, \,\, \forall n \in \mathcal{N}, f \in \mathcal{F}_{n},
    \end{equation} 
which consists of the actions conducted for the preceding~$T^\text{w}$ camera frames. We build a recurrent neural network with parameters~$\boldsymbol{\vartheta}^\text{S}$ to approximate policy~$\Pi$ by using the S-state~$\mathbf{s}^\text{S}_{n,f}$ as the input. In this case, the approximation of policy~$\Pi$ can be simplified as a temporal sequence prediction~\cite{hochreiter1997long}. We optimize parameters~$\boldsymbol{\vartheta}^\text{S}$ by minimizing the following loss function:
    \begin{equation}\label{}
        L(\boldsymbol{\vartheta}^\text{S}) = \frac{1}{|\Xi|} \sum_{(a_{n,f},\mathbf{s}^\text{S}_{n,f}) \in \Xi} \left(a_{n,f} - \Pi^\text{S}(\mathbf{s}^\text{S}_{n,f};{\boldsymbol{\vartheta}^\text{S}}) \right)^{2}.
    \end{equation}
Since S-state~$\mathbf{s}^\text{S}_{n,f}$ consists of only previous actions, S-state transitions are straightforward and do not require additional modeling.

Let~$N^\text{S}$ represent the number of cells of the adopted recurrent neural networks, e.g., LSTM~\cite{sak2014long}. Given that the input size and the output size of the recurrent neural network are~$T^\text{w}$ and~$1$, respectively, the computational complexity for predicting whether a camera frame will be uploaded using the S-model is $O(Z^\text{S})$, where~$Z^\text{S}$ is the total number of parameters in a standard recurrent neural network with one cell in each memory block and is given by~\cite{sak2014long}:
    \begin{equation}\label{}
        Z^\text{S} = 4 \times N^\text{S} \times N^\text{S} + 4 \times N^\text{S} \times (T^\text{w} + 1).
    \end{equation} 
Given~$F$ camera frames to be uploaded within the subsequent time slot, the computational complexity becomes~$O(Z^\text{S}F)$.

\subsection{Model Switching Function (MSF)}

    \begin{algorithm}[t] 
        \caption{Model Switching Method}\label{alg1}
        \LinesNumbered
        \textbf{Input:} $M_{n}$, $\delta^\text{u}_{n}$, $\delta^\text{b}_{n}$;\\
        \textbf{Initialization:} $h_{n,1} = 1, m_{n,1} = 0$;\\
        \For{$t \in \mathcal{T}$}
        {       
            Calculate~$\Delta_{n,t}$ by \eqref{eq9};\\
            \eIf{$\Delta_{n,t} > \delta^\text{u}_{n}$}
                {
                    $h_{n,t} \leftarrow 1$; $m_{n,t} \leftarrow 0$;\\
                }
                {  
                    \If{$\Delta_{n,t} < \delta^\text{b}_{n}$}
                    {
                        $m_{n,t} \leftarrow m_{n,t-1}+1$;\\
                        \If{$m_{n,t} \ge M_{n}$}
                        {
                           $h_{n,t} \leftarrow 0$; $m_{n,t} \leftarrow 0$;\\   
                        } 

                    }

                }
        }
        \textbf{Output:} $h_{n,t}$

    \end{algorithm} 

The model switching function (MSF) is designed to adapt the two data-driven models in the DMF to non-stationary uplink data traffic. We propose an easy-to-use model switching method stemmed from our observations from experiments conducted with an open-source SLAM-based device pose tracking platform~\cite{campos2021orb,ben2022edge} on real camera frame sequences~\cite{li2018interiornet}.

In MAR applications, when variations in the physical environment and device pose are insignificant, the procedure of key frame selection and uploading in the MAR operational mechanism is simple, leading to relatively stable key frame uploading; Conversely, a significant variation, e.g., tracking loss, generally complicates the procedure of key frame selection and uploading, potentially resulting in bursts of key frame uploading. We define the variation in the number of uploaded key frames across time slots, denoted by~$\Delta_{n,t}$, to guide the model switching for MAR device~$n$:
    \begin{equation}\label{eq9}
      \Delta_{n,t} =   \mid \tilde{k}_{n,t-1} - \tilde{k}_{n,t-2} \mid, \,\, \forall t \in \mathcal{T}, n \in \mathcal{N}.
    \end{equation}
We refer to~$\Delta_{n,t}$ as the \emph{model switching status} of MAR device~$n$. Define~$h_{n,t} \in \left\{0, 1\right\}$ as an indicator representing the resulting choice of our MSF for MAR device~$n$, i.e.,~its currently used modeling method. If~$h_{n,t} = 1$, the detailed modeling is used in time slot~$t$; Otherwise, the simplified modeling is used. 

The basic idea of model switching is as follows. A larger value of~$\Delta_{n,t}$ indicates a significant variation in device pose, and the detailed modeling is needed to capture the burst of key frame uploading in this case. Otherwise, the simplified modeling is sufficient. We present the model switching method in Algorithm~\ref{alg1}. We introduce parameters~$\delta^\text{u}_{n}$,~$\delta^\text{b}_{n}$, and $M_{n}$ to jointly configure the switching condition for MAR device~$n$, which avoid frequent switching between the two modeling methods. In lines~5-7, parameter~$\delta^\text{u}_{n}$ determines the threshold indicating that a significant variation occurs, and controls the switching to the detailed modeling, i.e.,~$h_{n,t} = 1$, when the value of~$\Delta_{n,t}$ exceeds the threshold. In lines 8-14, parameters~$\delta^\text{b}_{n}$ and~$M_{n}$ jointly control the switching to the simplified modeling, wherein $\delta^\text{b}_{n}$ represents the threshold for an insignificant variation. We introduce parameter~$M_{n}$ to adjust the trigger point for switching from the detailed modeling to the simplified modeling. A larger value of~$M_{n}$ represents a preference for using detailed modeling in the DMF. The three parameters offer flexibility for user-centric configuration, as they can be configured according to user movements and user-specific psychical environment.

\subsection{MAR User Profile (MUP)}\label{sec43}

    \begin{figure}[t]
        \centering
        \includegraphics[width=0.49\textwidth]{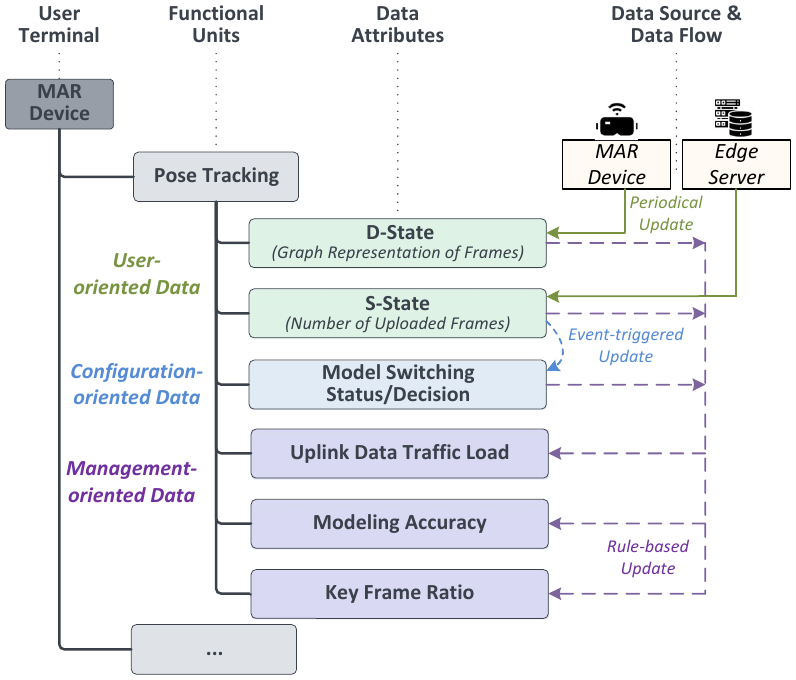}
        \caption{The hierarchical data model in the MUP.}\label{data_model}
    \end{figure}

The MUP includes a structured \emph{data model} consisting of a number of data elements, organized according to a defined schema~\cite{zhou2024user_wcm,ma2023nomore}. As shown in Fig.~\ref{data_model}, our defined hierarchical data model is tailored to MAR, which can support user-centric communication service provision. For our data model, we introduce data structure, data attribute, and data update method as follows:

    \begin{itemize}
        \item \textbf{Data Structure}: As shown in Fig.~\ref{data_model}, at the top level of this hierarchy, there is a ``user terminal'' representing an individual MAR device such as a pair of smart glasses or a smartphone. We define a number of ``functional units'' belonging to the MAR device, each relating to a unique functionality, e.g., device pose tracking or 3D rendering, in the MAR application. For each functional unit, e.g., device pose tracking, a set of ``data attributes'' can be flexibly defined according to its corresponding operation mechanism, thereby enabling MAR user-specific service demand modeling from a data-oriented perspective~\cite{zhou2024user_wcm}.

        \item \textbf{Data Attribute}: We classify the defined data attributes into three categories. (i) \emph{User-oriented data} support the characterization of service demand for the corresponding MAR device. Illustrated as the green blocks in Fig.~\ref{data_model}, key parameters~$\mathbf{s}^\text{D}_{n,f}$,~$\mathbf{s}^\text{S}_{n,f}$, and~$a_{n,f}$ are used by the DMF to model the uplink data traffic. (ii) \emph{Configuration-oriented data} support the effective configuration and execution of UDTOFs, e.g., DMF and MSF, thereby enabling the implementation of the DT itself. Model switching status~$\Delta_{n,t}$ and the choice of switching~$h_{n,t}$ are configured to determine the switching condition, shown as the blue block in Fig.~\ref{data_model}. (iii) \emph{Management-oriented data}, shown as the purple blocks in Fig.~\ref{data_model}, facilitate the implementation of user-centric network resource management strategies. In particular, parameters regarding demand modeling inaccuracies, such as~$p_{t}$, $q_{t}$, and $\lambda_{t}$, support the decision making on user-centric network resource allocation, which will be introduced in Section~\ref{sec5}.

        \item \textbf{Data Update}: The data elements in our data model can be updated. The data sources and data update methods can vary across data attributes. We design three different data update methods for the aforementioned three categories of data attributes. As shown in Fig.~\ref{data_model}, the user-oriented data are periodically collected from the MAR device and the edge server. The update of configuration-oriented data is triggered by the update events of user-oriented data. Specifically, the update of model switching status~$\Delta_{n,t}$ in Algorithm~\ref{alg1} is based on the temporal variation in~$a_{n,f}$, and the change of~$h_{n,t}$ is triggered by the value change of~$\Delta_{n,t}$. The update of the management-oriented data depends on an update rule tailored to the adopted network management algorithms, which will be presented in Section~\ref{sec5}.

    \end{itemize}

In contrast to conventional mathematical modeling for capturing service demand, our data model has two potential advantages. First, it has the flexibility to differentiate the data structure and the data attributes for different MAR devices, thereby supporting customized service demand modeling and network resource management. In addition, our data model has a high extensibility to support various functionalities in MAR applications and network management objectives.\footnote{This paper only focuses on the communication service provision for a single functional unit, i.e., device pose tracking.}  Due to these advantages, the implicit impact of the MAR operational mechanism on the user-specific service demands can be characterized by using the data model and UDTOFs~\cite{shen2021holistic,ma2025matterguard}.

\section{DT-based User-centric Communication Service Provision}\label{sec5}

It is well known that neither mathematical models nor data-driven models can guarantee accurate modeling of service demand~\cite{navarro2020survey,shen2021holistic}. Considering the potential modeling inaccuracies resulting from MARLIN, we need to design a robust communication service provision method tailored to data-driven traffic modeling.

Define~$\hat{a}_{n,f}$ as the predicted value of~$a_{n,f}$ by MARLIN. Given the prediction results~$\hat{\mathbf{a}}_{n,t} = [\hat{a}_{n, f}]_{\forall f \in \mathcal{F}_{n,t}} $, Problem~P1 can be transformed as follows:
    \begin{subequations}\label{p2}
        \begin{align}
            \textrm{P2:} &\,\, \min_{ \{ b_{n,t} \}_{n \in \mathcal{N}, t \in \mathcal{T}} } \sum_{n \in \mathcal{N} }{\sum_{ t \in \mathcal{T} }{ b_{n,t} }}\\
            \textrm{s.t.} &\,\, P( \alpha \tilde{k}_{n,t} \leq T^\text{r} r_{n,t} | \hat{\mathbf{a}}_{n,t} ) \ge \epsilon, \,\, \forall n \in \mathcal{N}, t \in \mathcal{T},
        \end{align}
    \end{subequations}
where constraint~(\ref{p2}b) bounds the minimum value of the conditional probability given the predicted results. We solve Problem~P2 to achieve robust user-centric communication service provision in both single-user and multi-user scenarios. 

\subsection{Single-user Communication Service Provision}\label{}

For the single-user scenario, i.e.,~$|\mathcal{N}| = 1$, in this subsection, we remove the subscript~$n$ for simplicity. To represent the amount of radio spectrum resource required for timely key frame uploading, we introduce an auxiliary variable~$k_{t}$, defined as the number of key frames if the reserved radio spectrum resource~$b_{t}$ can support timely key frame uploading. As a result, the optimal solution to Problem~P2, i.e.,~$b^{*}_{t}$, is given by:
    \begin{equation}\label{eq11}
      b^{*}_{t} = \frac{\alpha T^\text{r}}{\log (1 + \gamma_{t})}  k^{*}_{t},\,\, \forall t \in \mathcal{T},
    \end{equation}
where~$k^{*}_{t}$ denotes the minimum value of $k_{t}$, given by:
    \begin{equation}\label{eq12}
      k^{*}_{t} =  \arg \min_{k_{t}} P( \sum_{f \in \mathcal{F}_{t}}{a_{f} } \le k_{t} | \hat{\mathbf{a}}_{t} ) \ge \epsilon. 
    \end{equation}
To determine~$k^{*}_{t}$, we need to first obtain the conditional probability~$P(\mathbf{a}_{t} = \mathbf{a} | \hat{\mathbf{a}}_{t} = \hat{\mathbf{a}})$, i.e.,~the posterior probability of~$\mathbf{a}_{t} = \mathbf{a}, \forall \mathbf{a} \in \mathcal{A}$ given predicted results~$\hat{\mathbf{a}}_{t} = \hat{\mathbf{a}}, \forall \hat{\mathbf{a}} \in \mathcal{A}$ according to Bayes' theorem:
   \begin{equation}\label{eq13}
       P(\mathbf{a}_{t} = \mathbf{a} | \hat{\mathbf{a}}_{t} = \hat{\mathbf{a}}) = \frac{P(\mathbf{a}_{t} = \mathbf{a}) P(\hat{\mathbf{a}}_{t} = \hat{\mathbf{a}} | \mathbf{a}_{t} = \mathbf{a})}{\sum_{\mathbf{a} \in \mathcal{A}}{P(\mathbf{a}_{t} = \mathbf{a}) P(\hat{\mathbf{a}}_{t} = \hat{\mathbf{a}} | \mathbf{a}_{t} = \mathbf{a})} },
    \end{equation}
where~$P(\mathbf{a}_{t} = \mathbf{a})$ is the prior probability of~$\mathbf{a}_{t}$, and~$P(\hat{\mathbf{a}}_{t} = \hat{\mathbf{a}} | \mathbf{a}_{t} = \mathbf{a})$ is the likelihood of~$\mathbf{a}_{t}$ given prediction result~$\hat{\mathbf{a}}_{t}$. By updating the value of~$P(\mathbf{a}_{t} = \mathbf{a} | \hat{\mathbf{a}}_{t} = \hat{\mathbf{a}})$ based on the update of parameters~$P(\mathbf{a}_{t} = \mathbf{a})$ and~$P(\hat{\mathbf{a}}_{t} = \hat{\mathbf{a}} | \mathbf{a}_{t} = \mathbf{a})$, we can determine the value of~$k^{*}_{t}$. 

    \begin{algorithm}[t] 
        \caption{DT-based User-centric Communication Service Provision}\label{alg2}
        \LinesNumbered
        \textbf{Input:} $\tau$;\\
        \textbf{Initialization:} $P^\text{D}(\mathbf{a}_{t} = \mathbf{a}), P^\text{S}(\mathbf{a}_{t} = \mathbf{a}), P^\text{D}(\hat{\mathbf{a}}_{t} = \hat{\mathbf{a}} | \mathbf{a}_{t} = \mathbf{a}), P^\text{S}(\hat{\mathbf{a}}_{t} = \hat{\mathbf{a}} | \mathbf{a}_{t} = \mathbf{a}),\,\, \forall \mathbf{a} \in \mathcal{A}$;\\
        \For{$t \in \mathcal{T}$}
        {   
            
            The MSF obtains~$h_{t}$ by using Algorithm~\ref{alg1} and selects a method from the detailed and simplified modeling methods based on~$h_{t}$;\\

            The DMF obtains prediction result~$\hat{\mathbf{a}}_{t} = \hat{\mathbf{a}}$;\\

            \eIf{$h_{t} = 1$}
                {   
                    The MUP provides~$P^\text{D}(\mathbf{a}_{t} = \mathbf{a} | \hat{\mathbf{a}}_{t})$ by~\eqref{eq13};\\

                    The controller calculates $k^{*}_{t}$ by~\eqref{eq12} based on~$P^\text{D}(\mathbf{a}_{t} = \mathbf{a} | \hat{\mathbf{a}}_{t})$;\\ 
                }
                {   
                    The MUP provides~$P^\text{S}(\mathbf{a}_{t} = \mathbf{a} | \hat{\mathbf{a}}_{t})$ by~\eqref{eq13};\\

                    The controller calculates $k^{*}_{t}$ by~\eqref{eq12} based on~$P^\text{S}(\mathbf{a}_{t} = \mathbf{a} | \hat{\mathbf{a}}_{t})$;\\ 
                    
                }

            $\mathbf{s}^\text{D}_{f}, a_{f}\,\, \forall f \in \mathcal{F}_{t}$ are used to update the user-oriented data in the MUP;\\ 
            \eIf{$h_{t} = 1$}
                {   
                    $P^\text{D}(\hat{\mathbf{a}}_{t+1} = \hat{\mathbf{a}} | \mathbf{a}_{t+1} = \mathbf{a})$, $P^\text{D}(\mathbf{a}_{t+1} = \mathbf{a})$ $\leftarrow$ Estimate based on the user-oriented data regarding the detailed modeling during past~$\tau$ time slots;\\
                }
                {   
                    $P^\text{S}(\hat{\mathbf{a}}_{t+1} = \hat{\mathbf{a}} | \mathbf{a}_{t+1} = \mathbf{a})$, $P^\text{S}(\mathbf{a}_{t+1} = \mathbf{a})$ $\leftarrow$ Estimate based on the user-oriented data regarding the simplified modeling during past~$\tau$ time slots;\\
                    
                }

        }
        \textbf{Output:} $k^{*}_{t}$

    \end{algorithm} 

The workflow of DT-based user-centric communication service provision for the MAR device is summarized in Algorithm~\ref{alg2}. In lines~4-5, at the beginning of time slot~$t$, the MSF selects a proper modeling method, and the DMF provides the predicted results of the service demand of the MAR device in time slot~$t$. In lines~6-12, parameters~$P(\mathbf{a}_{t} = \mathbf{a})$ and~$P(\hat{\mathbf{a}}_{t} = \hat{\mathbf{a}} | \mathbf{a}_{t} = \mathbf{a})$ stored in the MUP are used to calculate~$P(\mathbf{a}_{t} = \mathbf{a} | \hat{\mathbf{a}}_{t} = \hat{\mathbf{a}})$, which can help the controller to determine the value of~$k^{*}_{t}$. To indicate the parameters for the detailed and the simplified modeling, we use superscripts ``$\text{D}$'' and ``$\text{S}$'', respectively. In line~13, information on D-states~$\mathbf{s}^\text{D}_{f}$ and actions~$a_{f}$, collected from the MAR device and the edge server, respectively, is used to update the user-oriented data in the MUP at the end of time slot~$t$. In lines~14-17, parameters~$P(\mathbf{a}_{t} = \mathbf{a})$ and~$P(\hat{\mathbf{a}}_{t} = \hat{\mathbf{a}} | \mathbf{a}_{t} = \mathbf{a})$ are updated based on the user-oriented data collected in the past~$\tau$ time slots. Parameters~$P(\mathbf{a}_{t} = \mathbf{a})$ and~$P(\hat{\mathbf{a}}_{t} = \hat{\mathbf{a}} | \mathbf{a}_{t} = \mathbf{a})$ in Algorithm~\ref{alg2} represent the management-oriented data defined in Subsection~\ref{sec43}, which enhance the robustness of decision making for user-centric communication service provision by accounting for modeling inaccuracies.

Considering the complexity in updating a large number of parameters~$P(\mathbf{a}_{t} = \mathbf{a})$ and~$P(\hat{\mathbf{a}}_{t} = \hat{\mathbf{a}} | \mathbf{a}_{t} = \mathbf{a})$ in the MUP, we analyze the performance of user-centric communication service provision and decrease the number of the used parameters for easily applying Algorithm~\ref{alg2}. Without loss of generality, we assume that, given~$\hat{a}_{f}, \forall f \in \mathcal{F}_{t}$, random variables~$a_{f}, \forall f \in \mathcal{F}_{t}$ are independent and identically distributed (i.i.d.), and~$P(a_{f} | \hat{\mathbf{a}}_{t}) = P(a_{f} | \hat{a}_{f}), \forall f \in \mathcal{F}_{t}$. 

    \begin{figure*}[t] 
        \begin{equation}\label{eq30}
            \begin{split}
                & g(k_{t};p_{t}, q_{t}, \lambda_{t}) = \sum_{k=0}^{k_{t}}{ \sum_{j = \max(0, k-(F_{t}-\hat{A}_{t}))}^{\min(\hat{A}_{t},k)}{\binom{\hat{A}_{t}}{j} (p_{t}^\text{TPR})^{j}(1-p_{t}^\text{TPR})^{\hat{A}_{t}-j} \binom{F_{t}-\hat{A}_{t}}{k-j} (1-p_{t}^\text{TNR})^{k-j}(p_{t}^\text{TNR})^{F_{t}-\hat{A}_{t}-k+j} }}.
            \end{split}
        \end{equation}  
        \rule[-10pt]{18.15cm}{0.05em} 
    \end{figure*}

    \begin{theorem}\label{theorem1}
        The probability~$P(\sum_{f \in \mathcal{F}_{t}}{a_{f}} \leq k_{t} | \hat{\mathbf{a}}_{t} )$ can be derived in~\eqref{eq30}, which is a non-decreasing function given parameters~$p_{t}$, $q_{t}$, and $\lambda_{t}$, where $\hat{A}_{t} = \sum_{f \in \mathcal{F}_{t}}{\hat{a}_{f}}$, $F_{t} = |\mathcal{F}_{t}|$,
            \begin{equation}\label{}
               p_{t}^\text{TPR} = \frac{p_{t}\lambda_{t}}{p_{t}\lambda_{t} +(1-q_{t})(1-\lambda_{t})},
            \end{equation}
        and
            \begin{equation}\label{}
               p_{t}^\text{TNR} = \frac{q_{t}(1-\lambda_{t})}{q_{t}(1-\lambda_{t}) + (1-p_{t})\lambda_{t}}.
            \end{equation}
    \end{theorem}
    \begin{proof}   
        See Appendix~A.
    \end{proof}

Theorem~\ref{theorem1} allows us to find the optimal value, i.e.,~$k^{*}_{t}$, based on the closed-form probability~$P(\sum_{f \in \mathcal{F}_{t}}{a_{f}} \leq k_{t} | \hat{\mathbf{a}}_{t} )$, given only three parameters~$p_{t} = P(\hat{a}_{f}=1 | a_{f}=1)$,~$q_{t} = P(\hat{a}_{f}=0 | a_{f}=0)$, and~$\lambda_{t} = P(a_{f}=1)$ in time slot~$t$. Thus, parameters~$p_{t}$, $q_{t}$, and $\lambda_{t}$ can replace parameters~$P(\mathbf{a}_{t} = \mathbf{a})$ and~$P(\hat{\mathbf{a}}_{t} = \hat{\mathbf{a}} | \mathbf{a}_{t} = \mathbf{a})$ in Algorithm~\ref{alg2} as management-oriented data maintained by the MUP. By updating the three parameters, MARLIN supports robust user-centric communication service provision for the MAR device.

\subsection{Multi-user Communication Service Provision}\label{}

The proposed Algorithm~\ref{alg2} can be extended to multi-user scenarios, i.e.,~$|\mathcal{N}|> 1$, since the communication service provision for different MAR devices in Problem~P2 is independent except for a potential constraint on the total available radio spectrum resource for concurrent uplink transmissions at the BS. In this subsection, we focus on the comparison of user-centric communication service provision and slicing-based communication service provision.

Given the radio spectrum resource reserved for MAR device~$n$ in time slot~$t$, i.e.,~$b^{*}_{n,t}$, in~\eqref{eq11}, the overall radio spectrum reserved for all MAR devices in time slot~$t$ is given by: 
    \begin{equation}\label{eq17}
      B^{*}_{t} = \sum_{n \in \mathcal{N}}{ \frac{\alpha T^\text{r}}{\log (1 + \gamma_{n,t})}  k^{*}_{n,t} },\,\, \forall t \in \mathcal{T}.
    \end{equation}
In user-centric communication service provision, the reserved radio spectrum~$B^{*}_{t}$ can ensure that the probability of timely uploading key frame for each MAR device is larger than the value of~$\epsilon$.

Since slicing-based communication service provision does not have or use the service demand information of individual MAR devices~\cite{shen2021holistic}, the service demands of different devices running the same MAR applications are usually assumed to be i.i.d.~\cite{ye2018dynamic}. Considering a simple scenario where random variables~$\tilde{k}_{n,t}, \forall n  \in \mathcal{N}$ are i.i.d. with mean~$\bar{k}_{t}$ and variance~$\sigma^{2}_{t}$, we can formulate a slicing-based communication service provision problem as follows~\cite{li2022risk,li2020data}:
    \begin{subequations}\label{p3} 
        \begin{align}
            \textrm{P3:} &\,\, \min_{ \{ B^\text{s}_{t} \}_{t \in \mathcal{T}} } \sum_{ t \in \mathcal{T} }{ B^\text{s}_{t} }\\
            \textrm{s.t.} &\,\, P \left( \frac{\alpha}{|\mathcal{N}|} \sum_{n \in \mathcal{N}}{\tilde{k}_{n,t}} \leq T^\text{r} r\left( \frac{B^\text{s}_{t}}{|\mathcal{N}|} \right) \right) \ge \epsilon, \,\, \forall t \in \mathcal{T},
        \end{align}
    \end{subequations}
where~$B^\text{s}_{t}$ represents the radio spectrum resource reserved for all MAR devices, and~(\ref{p3}b) ensures that the probability of timely uploading the number of key frames averaged over all MAR devices is larger than the value of~$\epsilon$. 

    \begin{theorem}\label{theorem2}
        The optimal value of~$B^\text{s}_{t}$ in Problem~P3 is given by:
            \begin{equation}\label{eq19}
               B^{s,*}_{t}= \frac{\alpha T^\text{r} |\mathcal{N}|}{\log (1 + \bar{\gamma}_{t})} \lceil \bar{k}_{t} + \frac{\Phi^{-1}(\epsilon) \sigma^{2}_{t}}{|\mathcal{N}|} \rceil,
            \end{equation}
        where $\Phi^{-1}(\cdot)$ denotes the quantile function of the standard normal distribution, $\bar{\gamma}_{t}$ represents the average SNR over all MAR devices, and~$ \lceil \cdot \rceil$ is the ceiling function. 
    \end{theorem}
    \begin{proof}   
        See Appendix~B.
    \end{proof}

Theorem~\ref{theorem2} allows us to obtain the overall radio spectrum reserved for MAR devices in time slot~$t$, i.e.,~$B^{s,*}_{t}$, in slicing-based communication service provision. Comparing~\eqref{eq17} and~\eqref{eq19}, we have two observations. First, the amount of radio spectrum resource reserved for an individual MAR device in the slicing-based communication service provision, i.e.,~$B^{s,*}_{t}/|\mathcal{N}|$, approaches~$\alpha T^\text{r} \bar{k}_{t} / \log (1 + \bar{\gamma}_{t})$, which is less than that in user-centric communication service provision when~$ \sum_{n \in \mathcal{N}}{k^{*}_{n,t}} < \bar{k}_{t} |\mathcal{N}|$. Second, slicing-based communication service provision is less sensitive to the variances of individual traffic demands of MAR devices. As shown in Theorem~\ref{theorem2}, the value of~$\Phi^{-1}(\epsilon) \sigma^{2}_{t}/|\mathcal{N}|$ decreases with the increase of~$|\mathcal{N}|$. In contrast, the distribution of~$\tilde{k}_{n,t}$ significantly influences user-centric communication service provision, and, thus, the user-centric communication service provision is more robust in adapting to user-specific service demands. Consequently, the user-centric communication service provision for MAR can trade off resource consumption for the ability in adapting to user-specific service demand variations.

\section{Performance Evaluation}\label{sec6}

In this section, we evaluate the performance of our DT-based, user-centric approach to communication service provision in edge-assisted MAR. Our focus is on the impact of radio spectrum resource management on the timeliness of key frame uploading from the communication and networking perspective.

\subsection{Simulation Settings}

\subsubsection{Settings of Considered Scenario and DT}
In our simulation, we use 218 camera frame sequences, corresponding to different user movements in various environments, from the InteriorNet dataset~\cite{li2018interiornet} and conduct device pose tracking for the MAR device using the open-source ORB-SLAM3 platform~\cite{campos2021orb}. We use a resource block (RB) as the base unit for radio spectrum resource, each of which is 180\,kHz wide (12 subcarriers) in radio spectrum bandwidth and 0.5\,ms in time. Other important parameter settings are listed in Table~I. 

To establish a DT for each MAR device, we first design the DMF in the DT using deep learning techniques. Specifically, for detailed modeling within the DMF, we employ a DNN with parameters~$\boldsymbol{\vartheta}^\text{D}$. This network consists of one graph convolution (GC) layer for graph feature extraction, followed by three fully connected (FC) layers with output dimensions of 416, 64, and 8, respectively. To approximate D-state transitions, we use another DNN with parameters~$\boldsymbol{\theta}$, comprising three GC layers followed by three FC layers with output dimensions of 256, 64, and 32, respectively. The input node features of graphs~$\mathcal{G}_{n,f}$ and~$\mathcal{G}^\text{D}_{n,f}$ are derived from camera frames processed by the ResNet-50 DNN. Each GC or FC layer is followed by a batch normalization layer and a ReLU activation function. For simplified modeling in the DMF, we adopt a DNN with parameters~$\boldsymbol{\vartheta}^\text{S}$, which includes one LSTM layer with 50 neurons, followed by two FC layers with output dimensions of 32 and 16, respectively.

Second, we design the MSF in the DT by concentrating on three key parameters that enable the adaptation of the DMF to the non-stationary uplink data traffic of each MAR device. We set the trigger point~$M_{n}$ and the lower variation threshold~$\delta^\text{b}_{n}$ to~3 and~2, respectively, to jointly determine the conditions for switching from detailed modeling to simplified modeling in the DMF. We set the upper variation threshold~$\delta^\text{u}_{n}$ to~4 to determine determine the condition for switching from simplified modeling back to detailed modeling.

\subsubsection{Benchmark Approaches}
We compare the proposed DT-based user-centric communication service provision with the state-of-the-art slicing-based communication service provision, a key innovation for service-oriented network management in 5G~\cite{ye2018dynamic}. In addition, we adopt the following prevalent service demand modeling approaches, including mathematical model-based and data-driven approaches, as benchmark~\cite{wang2020long,navarro2020survey,oreshkinn,zhou2021informer}:
    \begin{itemize}
        \item \emph{Poisson regression}: The number of key frames for uploading in each time slot is assumed to follow a Poisson distribution, a mathematical model widely applied in network resource management. The parameters of the Poisson distribution are estimated based on historical information.

        \item \emph{LSTM-based neural network}: Following the simplified modeling in the DMF, an LSTM-based recurrent neural network is pre-trained and employed to predict the number of key frames to be uploaded in each time slot.

        \item \emph{N-BEATS:} Without relying on recurrent neural networks or attention-based mechanisms, N-BEATS uses a stack of fully-connected blocks to capture long-term components for time-series prediction.

        \item \emph{Transformer-based neural network:} Transformer-based neural networks leverage self-attention mechanisms to enable parallel processing and capture direct relationships across all inputs, making them highly effective for time-series analysis and prediction.

        \item \emph{Informer:} Transformer-based neural networks designed for long sequence time-series forecasting, featuring a compute-efficient self-attention mechanism and a scalable architecture to handle long outputs.
  
    \end{itemize}

    \begin{table}[t]
        \normalsize
        \centering
        \captionsetup{justification=centering,singlelinecheck=false}
        \caption{Simulation Parameters}\label{table1}
        \begin{tabular}{c|c|c|c}
            \hline\hline
             Parameter & Value & Parameter & Value\\
             \hline\hline
             $F$ & 10\,frames & $T^\text{r}$ & 0.02\,second \\
             \hline
             $\alpha$ & 5\,Mbits & $\gamma_{n,t}$ & 15\,dB\\
             \hline
        \end{tabular}
    \end{table}

\subsection{Performance of Service Demand Modeling}

    \begin{figure}[t]
        \centering
        \includegraphics[width=0.40\textwidth]{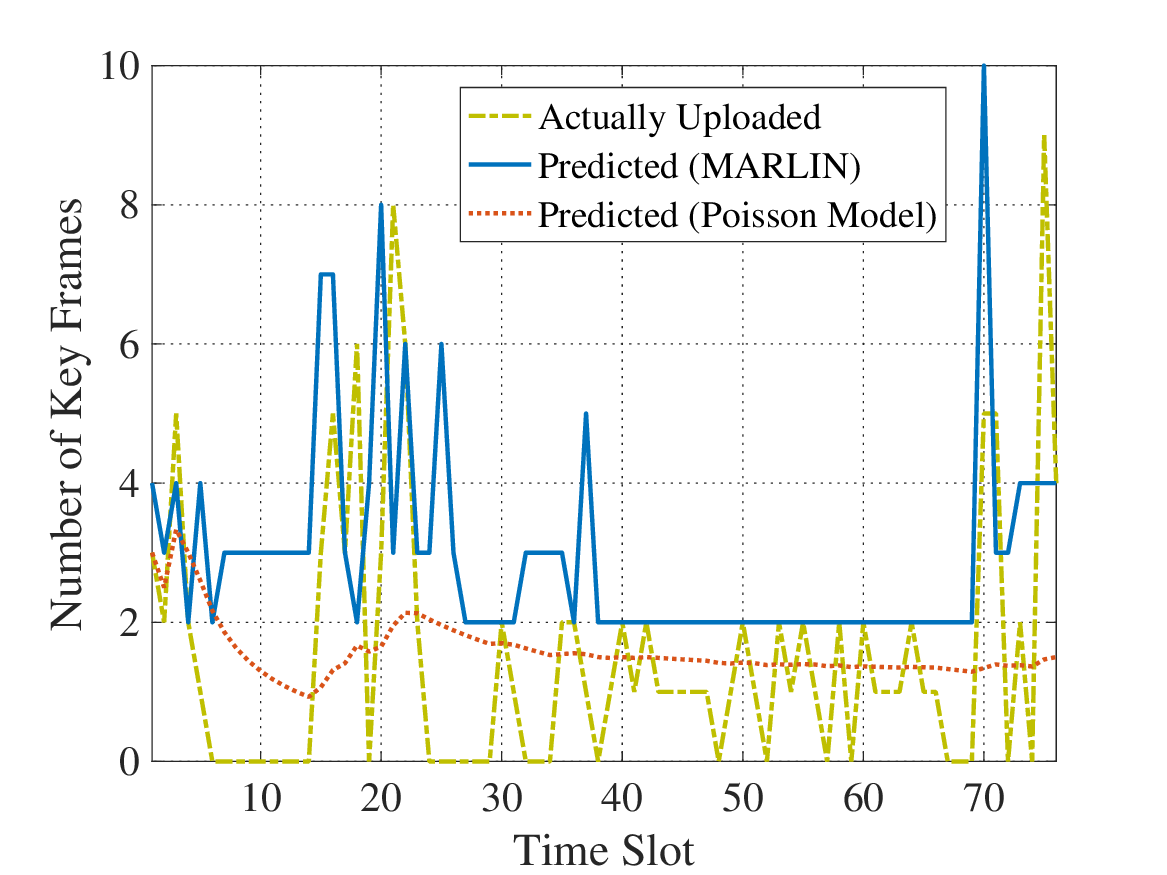}
        \caption{Performance Comparison between MARLIN and Poisson model-based approach.}\label{figa1}
    \end{figure}
We first compare the service demand modeling performance of MARLIN with that of a mathematical model-based approach. As shown in Fig.~\ref{figa1}, compared to the ``Predicted (Poisson Model)'' approach, we observe that the predicted values by MARLIN more closely match the actual non-stationary uplink data traffic, particularly during bursts in uplink data traffic that may result from device tracking loss or significant changes in the physical environment. In contrast, the widely-used Poisson model captures only large time-scale service demand patterns, such as average values, but falls short in characterizing the detailed temporal patterns in service demand, e.g., fluctuations caused by user movements or the SLAM-specific camera frame uploading mechanism. Therefore, MARLIN outperforms the Poisson model-based approach in terms of representation power. In addition, MARLIN can switch between the detailed and the simplified data-driven modeling according to variations in the number of uploaded key frames, thereby reducing input data redundancy in the detailed modeling.

    \begin{figure}[t]
        \centering
        \begin{subfigure}[b]{0.40\textwidth}
            \includegraphics[width=\textwidth]{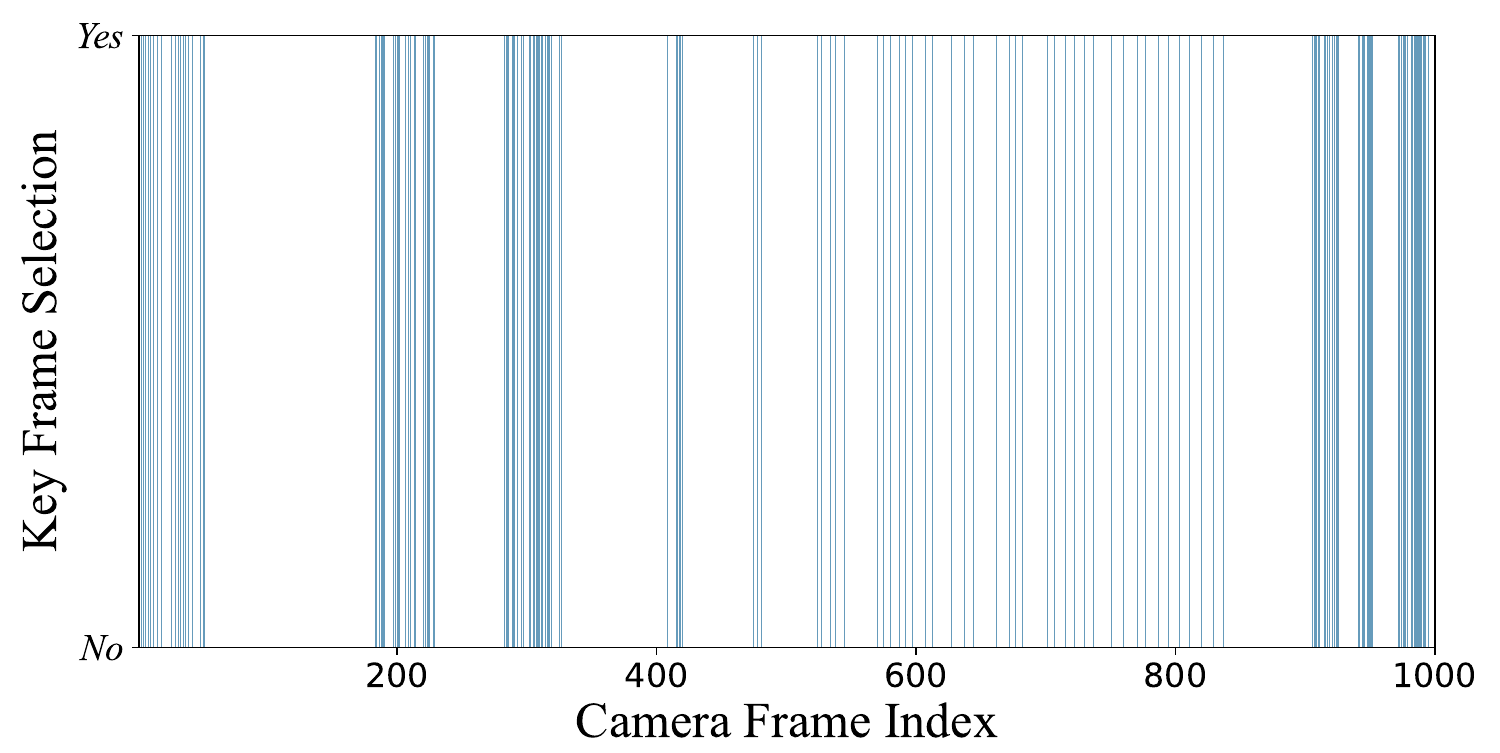}
            \caption{Non-stationary uplink data traffic due to the MAR operational mechanism of device pose tracking.}
            \label{figa2_1}
        \end{subfigure}
        \quad
        \begin{subfigure}[b]{0.40\textwidth}
            \includegraphics[width=\textwidth]{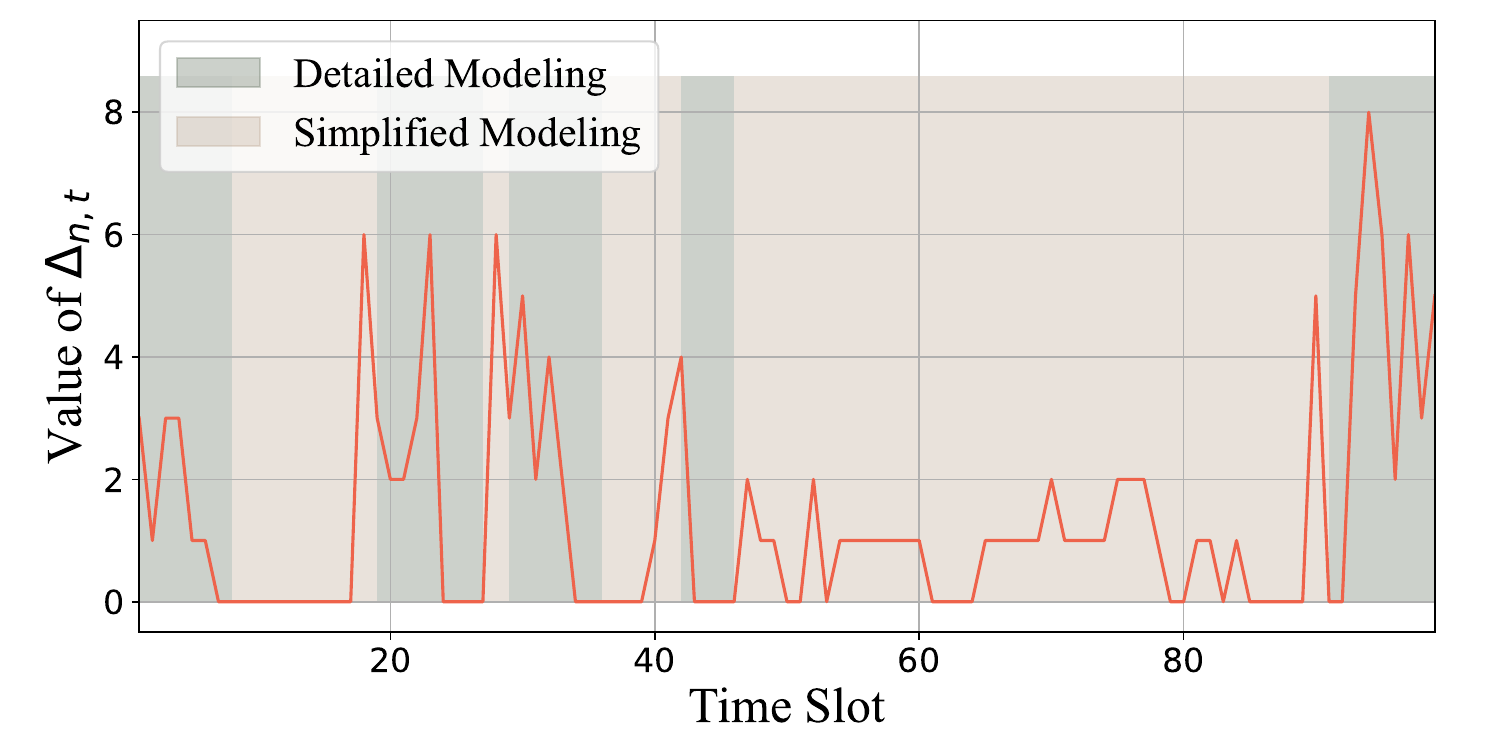}
            \caption{The change of model switching status.}
            \label{figa2_2}
        \end{subfigure}
        \caption{Performance of model switching in adapting to data traffic variations.}
        \label{figa2}
    \end{figure}

In Fig.~\ref{figa2}, we show the adaptivity of MARLIN to the non-stationary uplink data traffic for key frame uploading. For a camera sequence, the status whether each camera frame is selected as a key frame for uploading or not is shown in Fig.~\ref{figa2}(a). It is hard to identify a regular and stationary pattern at first glance since the complex operational mechanism of device pose tracking is adopted to address uncertain user movements. We plot the variations in model switching status, i.e.~$\Delta_{n,t}$, along with the corresponding choice of the detailed and the simplified modeling methods made by our MSF, in Fig.~\ref{figa2}(b). We can observe that the detailed modeling method is used during bursts in uplink data traffic, while the simplified modeling method is used in scenarios with less variations in uplink data traffic. This is because MARLIN targets accurately approximating the non-stationary MAR service demand by using contextual knowledge,~e.g., the MAR operational mechanism of device pose tracking in dealing with uncertain user movements. 

\subsection{Performance of DT-based User-centric Communication Service Provision}\label{sec63}

In this subsection, we compare the performance of our proposed approach with that of both mathematical model-based and data-driven approaches across 15 camera sequences. Since our objective is to minimize the communication resource consumption while ensuring timely camera frame uploading, we evaluate all approaches from two perspectives: 1) the ratio of timely uploaded key frames (TUKF) and 2) the radio spectrum resource provisioning (RP) ratio, defined as the amount of provisioned radio spectrum resource over the amount of required radio spectrum resource.

    \begin{figure}[t]
        \centering
        \includegraphics[width=0.40\textwidth]{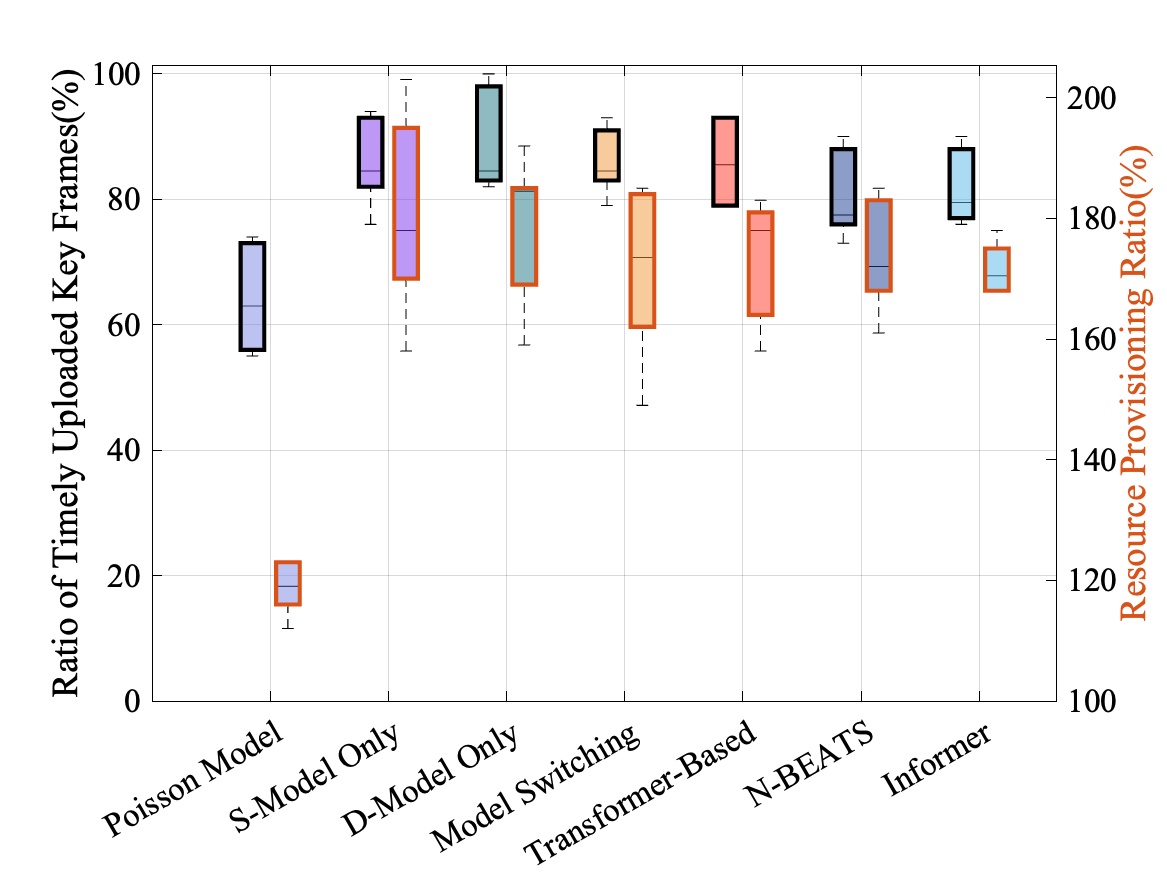}
        \caption{Radio spectrum resource utilization of the designed model switching method in MARLIN.}\label{figa3}
    \end{figure}

In Fig.~\ref{figa3}, we show the performance when only employing the simplified model (labeled as ``S-Model Only''), the performance when only employing the detailed model (labeled as ``D-Model Only''), and the performance when switching between two modeling methods (labeled as ``Model Switching''), respectively. We can make two important observations. First, ``Poisson Model'' is inferior to ``S-Model Only'', ``D-Model Only'', ``Transformer-based'', ``N-BEATS'', and ``Informer'' in both aspects since the data-driven modeling methods based on DNNs offer greater representation power. Second, compared to other data-driven benchmark methods, our approach achieves lower RP ratio without sacrificing the performance of TUKF significantly. This demonstrates the advantage of our approach in communication resource efficiency since the S- and D-models can be adaptively selected in different scenarios via the MSF. Compared to using only a data-driven modeling method, a hidden benefit of ``Model Switching'' lies in its flexibility to balance the modeling accuracy and the inference costs associated with establishing and updating the models in the MUP, by judiciously selecting from the lightweight S-model and the compute-intensive D-model.

    \begin{table*}[t]
      \centering
      \captionsetup{justification=centering,singlelinecheck=false}
      \caption{Comparison Results on Camera Frame Sequences with Different User Movement Patterns.}\label{table2}
      \begin{tabular}{c|cc|cc|cc|cc|cc}
      \hline\hline
      \multirow{2}{*}{\textbf{Approach}} & \multicolumn{2}{c|}{\begin{tabular}[c]{@{}c@{}}\textbf{3FO4K7I2Q0PG}\\ \textbf{original\_1\_1}\end{tabular}} & \multicolumn{2}{c|}{\begin{tabular}[c]{@{}c@{}}\textbf{3FO4K8S8DIKV}\\ \textbf{original\_3\_3}\end{tabular}} & \multicolumn{2}{c|}{\begin{tabular}[c]{@{}c@{}}\textbf{3FO4JWLHP0T1}\\ \textbf{original\_7\_7}\end{tabular}} & \multicolumn{2}{c|}{\begin{tabular}[c]{@{}c@{}}\textbf{3FO4JU3UAANP}\\ \textbf{random\_3\_3}\end{tabular}} & \multicolumn{2}{c}{\begin{tabular}[c]{@{}c@{}}\textbf{3FO4JTFDJQ38}\\ \textbf{random\_7\_7}\end{tabular}} \\ \cline{2-11} 
                              & \multicolumn{1}{c|}{TUKF(\%)}                           & RP(\%)                           & \multicolumn{1}{c|}{TUKF(\%)}                           & RP(\%)                           & \multicolumn{1}{c|}{TUKF(\%)}                           & RP(\%)                           & \multicolumn{1}{c|}{TUKF(\%)}                          & RP(\%)                          & \multicolumn{1}{c|}{TUKF(\%)}                          & RP(\%)                          \\ \hline
      \textbf{Our Proposed}                & \multicolumn{1}{c|}{91}                                 & 171                              & \multicolumn{1}{c|}{93}                                 & 166                              & \multicolumn{1}{c|}{83}                                 & 162                              & \multicolumn{1}{c|}{86}                                & 182                             & \multicolumn{1}{c|}{83}                                & 188                             \\ \hline
      \textbf{Poisson Model-based}           & \multicolumn{1}{c|}{56}                                 & 116                              & \multicolumn{1}{c|}{59}                                 & 116                              & \multicolumn{1}{c|}{73}                                 & 138                              & \multicolumn{1}{c|}{55}                                & 112                             & \multicolumn{1}{c|}{74}                                & 122                             \\ \hline
      \textbf{LSTM-based}             & \multicolumn{1}{c|}{94}                                 & 182                              & \multicolumn{1}{c|}{91}                                 & 195                              & \multicolumn{1}{c|}{86}                                 & 170                              & \multicolumn{1}{c|}{82}                                & 176                             & \multicolumn{1}{c|}{83}                                & 203                             \\ \hline
      \textbf{N-BEATS}                 & \multicolumn{1}{c|}{91}                                 & 185                              & \multicolumn{1}{c|}{91}                                 & 179                              & \multicolumn{1}{c|}{76}                                 & 161                              & \multicolumn{1}{c|}{82}                                & 173                             & \multicolumn{1}{c|}{79}                                & 168                             \\ \hline
      \textbf{Transformer-based}       & \multicolumn{1}{c|}{92}                                 & 183                              & \multicolumn{1}{c|}{91}                                 & 174                              & \multicolumn{1}{c|}{79}                                 & 176                              & \multicolumn{1}{c|}{85}                                & 184                             & \multicolumn{1}{c|}{81}                                & 187                             \\ \hline
      \textbf{Informer}                         & \multicolumn{1}{c|}{88}                                 & 175                              & \multicolumn{1}{c|}{90}                                 & 178                              & \multicolumn{1}{c|}{76}                                 & 160                              & \multicolumn{1}{c|}{79}                                & 169                             & \multicolumn{1}{c|}{77}                                & 176                             \\ \hline\hline
      \end{tabular}
    \end{table*}

In Table~\ref{table2}, we show the performance of these communication service provision approaches on different sub-datasets, i.e., camera frame sequences. In the InteriorNet dataset~\cite{li2018interiornet}, camera frame sequences are broadly classified into ``original'' and ``random'' categories. The former is generated based on human user movement patterns, while the latter is a randomly generated continuous trajectory, where``1-1'', ``3-3'', and ``7-7'' represent different user movement patterns, including velocity and acceleration. In terms of TUKF, we observe that some benchmark approaches may outperform our proposed approach in certain scenarios. For example, the LSTM-based and Transformer-based approaches achieve 3\% and 1\% higher TUKF, respectively, on the ``original\_1\_1'' sub-dataset, enabling more timely uploading of key frames. However, this improvement comes at the cost of approximately 12\% higher radio spectrum resource consumption. This observation indicates that our approach achieves a more efficient balance between TUKF and RP by dynamically switching between different data-driven models, demonstrating superior adaptability to non-stationary uplink data traffic. Another observation is that the performance gain of our proposed approach on the ``original'' sub-datasets is greater than that on the ``random'' sub-datasets. Our hypothesis is that the ``original'' sub-datasets demonstrate phased fluctuations in human user movement patterns, as shown in Fig.~\ref{figa2}(a), in which the proposed ``Model Switching'' excels. This demonstrates the potential of MARLIN to capture user-specific service demand in future 6G networks.

    \begin{figure}[t]
        \centering
        \includegraphics[width=0.40\textwidth]{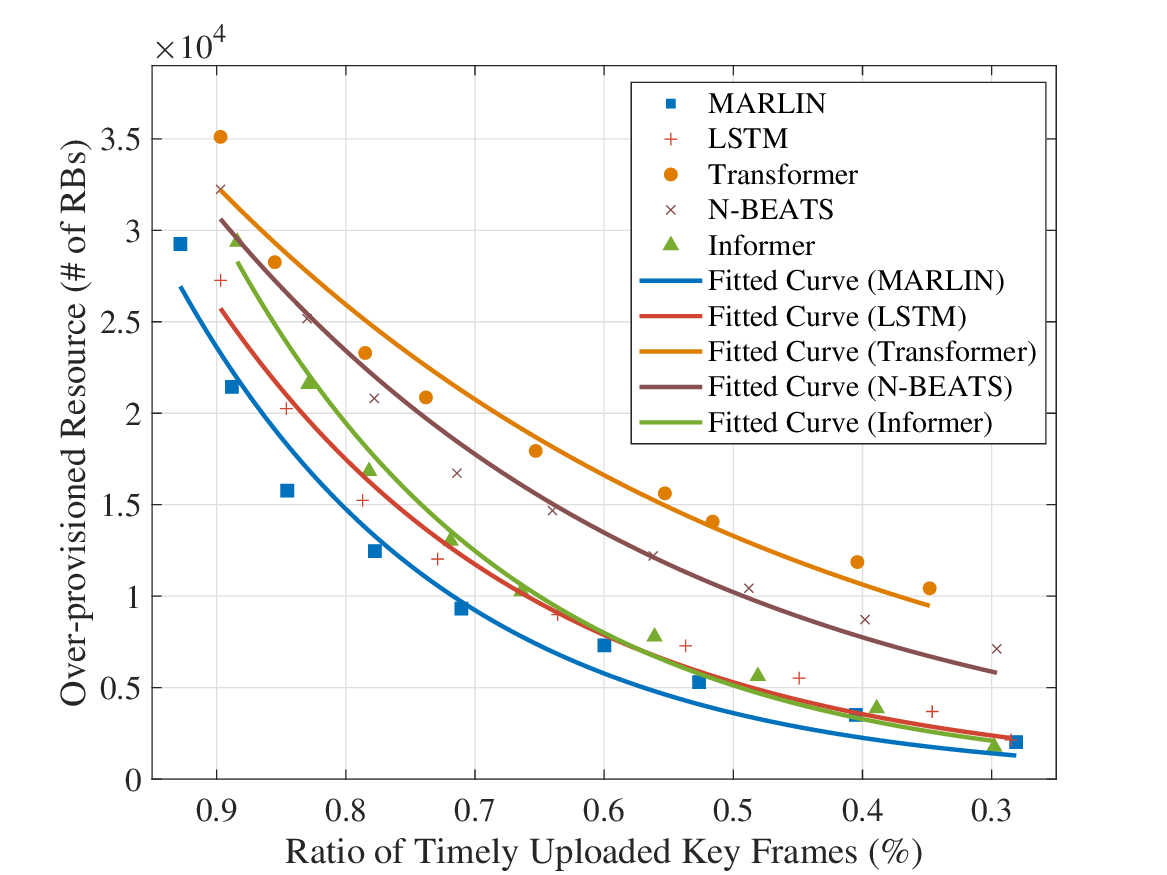}
        \caption{Performance comparison between the proposed DT-based approach and benchmark data-driven approaches.}\label{figb1}
    \end{figure}

In Fig.~\ref{figb1}, we compare the communication service provision performance of the DT-based approach with hat of benchmark data-driven approaches (labeled as ``LSTM'', ``Transformer'', ``N-BEATS'', and ``Informer'') in terms of the amount of over-provisioned radio spectrum resource (in RBs) and TUKF. We can observe that our DT-based approach not only reduces over-provisioning but also ensures the timeliness of key frame uploading of the MAR device, for two reasons. First, switching between different data-driven models can effectively capture the joint impact of the complex SLAM-based camera frame uploading mechanism and uncertain user movements on user-specific service demands. Second, our user-centric communication service provision, guided by theoretical analysis benefit from advanced data-driven techniques while accounting for data-driven modeling inaccuracies, thereby effectively mitigating the over-provisioning of radio spectrum resource. 

    \begin{figure}[t]
        \centering
        \includegraphics[width=0.40\textwidth]{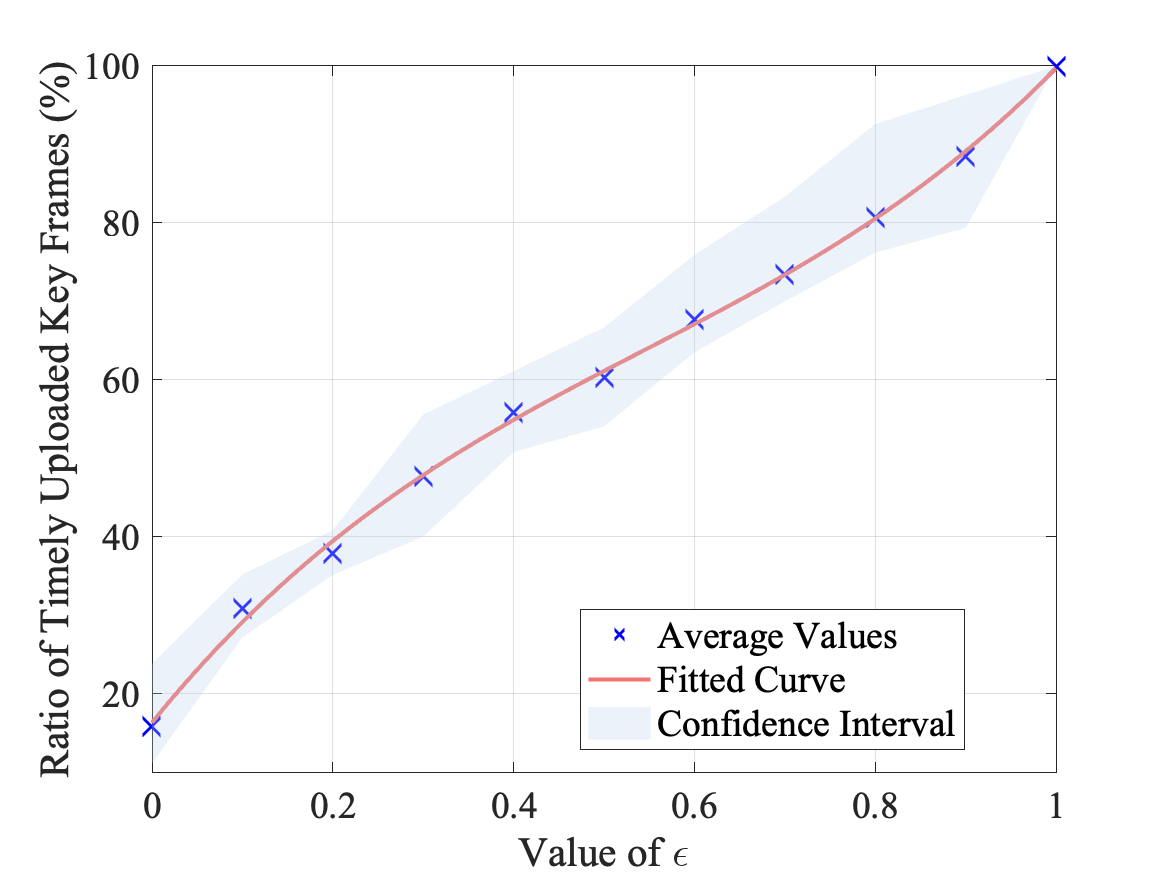}
        \caption{The value of~$\epsilon$ versus the actual ratio of timely uploaded key frames.}\label{figb3}
    \end{figure}

In Fig.~\ref{figb3}, given different values of~$\epsilon$, we plot the actual ratios of timely uploaded key frames with blue markers, each of which is averaged over 6 different camera sequences. The red curve represents a fitted curve, and the light blue shaded area represents the confidence interval between the minimum and maximum values. Given any~$\epsilon$, the radio spectrum resource is reserved for the MAR device based on~\eqref{eq30}. We can observe that the actual ratio of timely uploaded key frames closely approaches the value of~$\epsilon$, despite a slight difference. This is because the result in~\eqref{eq30} is based on the assumption that, given~$\hat{a}_{f}, \forall f \in \mathcal{F}_{t}$, random variables~$a_{f}, \forall f \in \mathcal{F}_{t}$ are i.i.d., and~$P(a_{f} | \hat{\mathbf{a}}_{t}) = P(a_{f} | \hat{a}_{f}), \forall f \in \mathcal{F}_{t}$. However, when the assumption holds, the cost of updating management-oriented data can be reduced. MARLIN can use only three parameters, i.e.,~$p_{t}$, $q_{t}$, and $\lambda_{t}$, as management-oriented data to conduct robust communication service provision with a slight difference in the performance of communication service provision.

    \begin{figure}[t]
        \centering
        \includegraphics[width=0.40\textwidth]{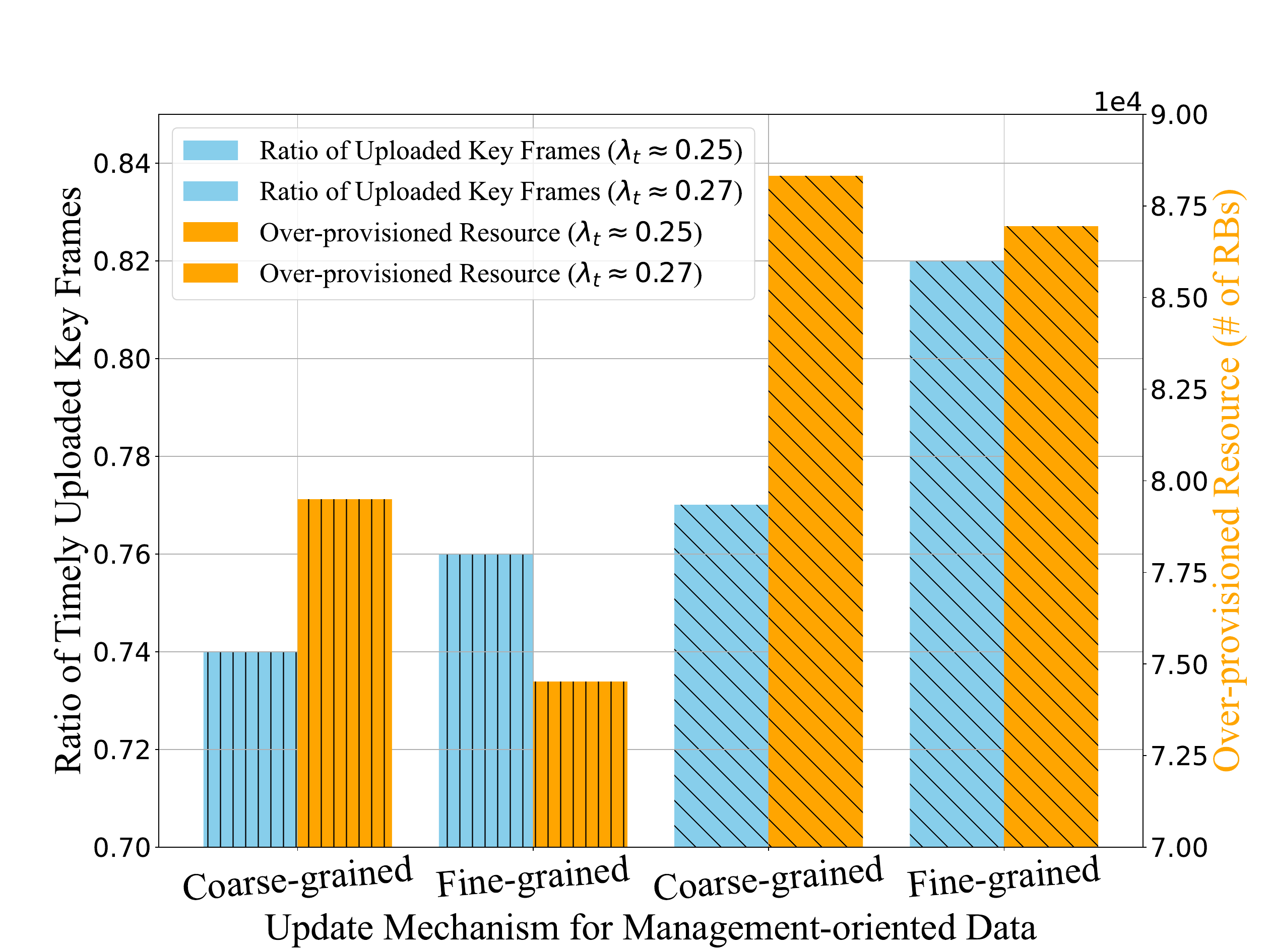}
        \caption{Performance of the fine-grained and the coarse-grained data update methods in MARLIN.}\label{figb2}
    \end{figure}

In Fig.~\ref{figb2}, we compare the performance of fine-grained and coarse-grained data management in MARLIN. As mentioned in~Algorithm~\ref{alg2}, we can separately update the management-oriented data (labeled as ``fine-grained''), e.g.~$p_{t}$, $q_{t}$, and $\lambda_{t}$, for the S-model and D-model. In contrast, we do not differentiate the collection of management-oriented data for the S-model and D-model for the benchmark (labeled as ``coarse-grained''). We observe that, given two different values of~$\lambda_{t}$, the ``fine-grained'' method outperforms the ``coarse-grained'' method in terms of the ratio of timely transmitted key frames and the amount of over-provisioned radio spectrum resource. This is because the ``fine-grained'' data update method empowers user-centric communication service provision by accurately characterizing the potential modeling inaccuracies of the detailed and the simplified modeling methods, especially when the two modeling methods suit different user movement patterns.

      \begin{figure*}[t]
        \centering
        \begin{subfigure}[b]{0.32\textwidth}
            \includegraphics[width=\textwidth]{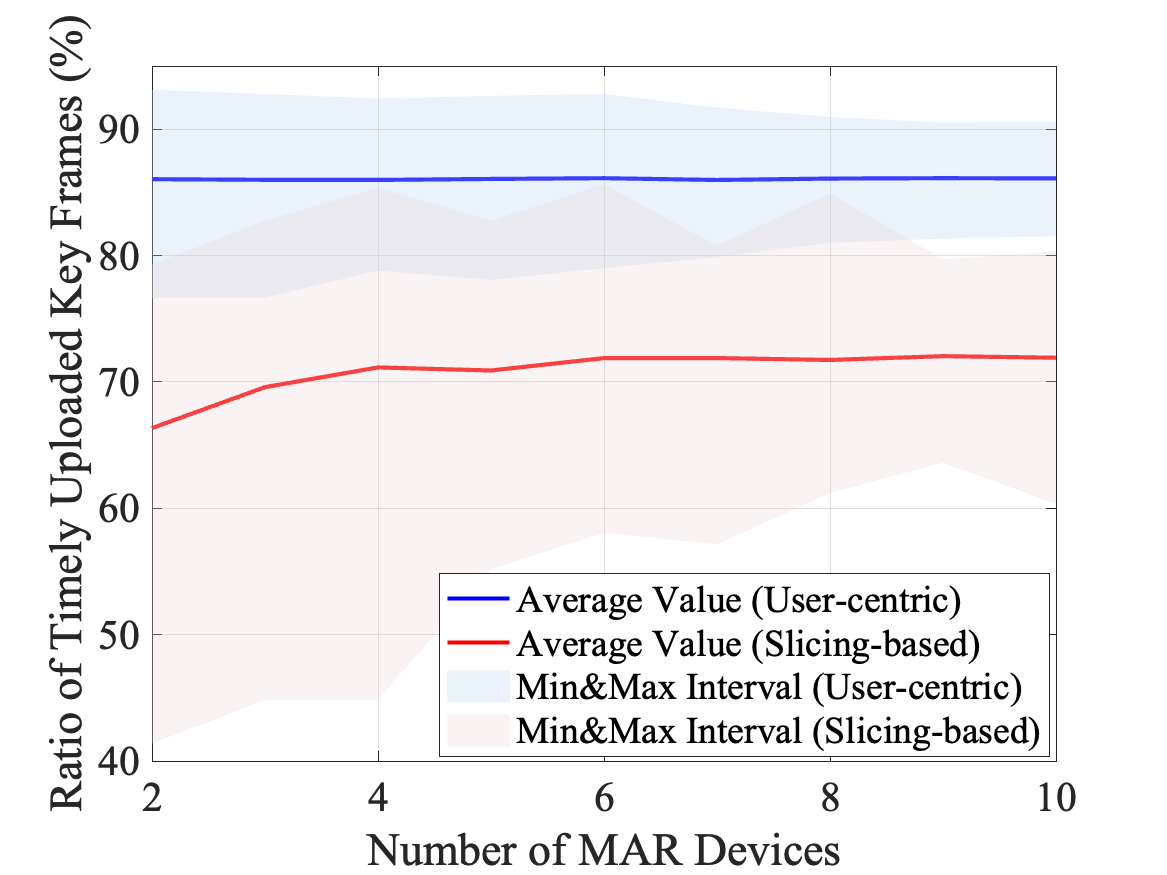}
            \caption{The ratio of timely uploaded key frames versus the number of MAR devices.}
            \label{figc1}
        \end{subfigure}
        \quad
        \begin{subfigure}[b]{0.32\textwidth}
            \includegraphics[width=\textwidth]{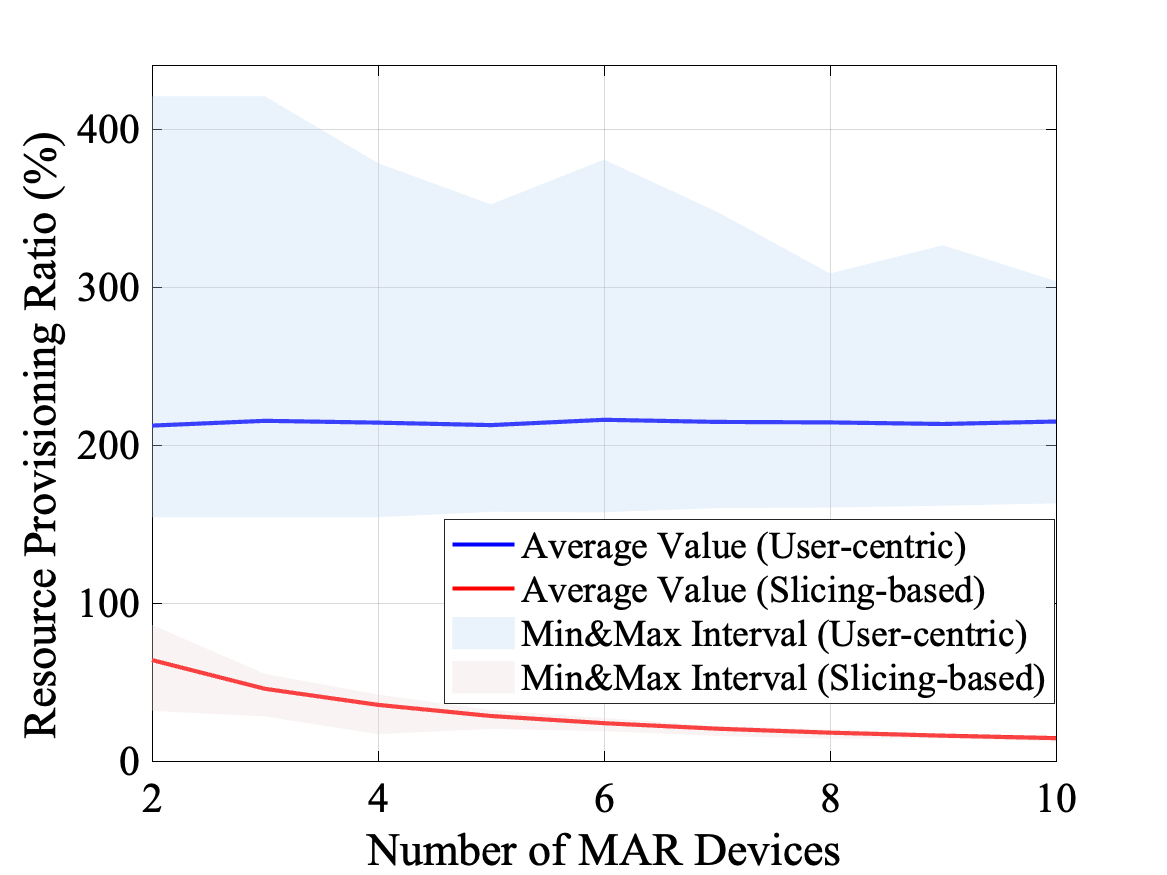}
            \caption{The resource provisioning ratio versus the number of MAR devices.}
            \label{figc2}
        \end{subfigure}
        \begin{subfigure}[b]{0.32\textwidth}
            \includegraphics[width=\textwidth]{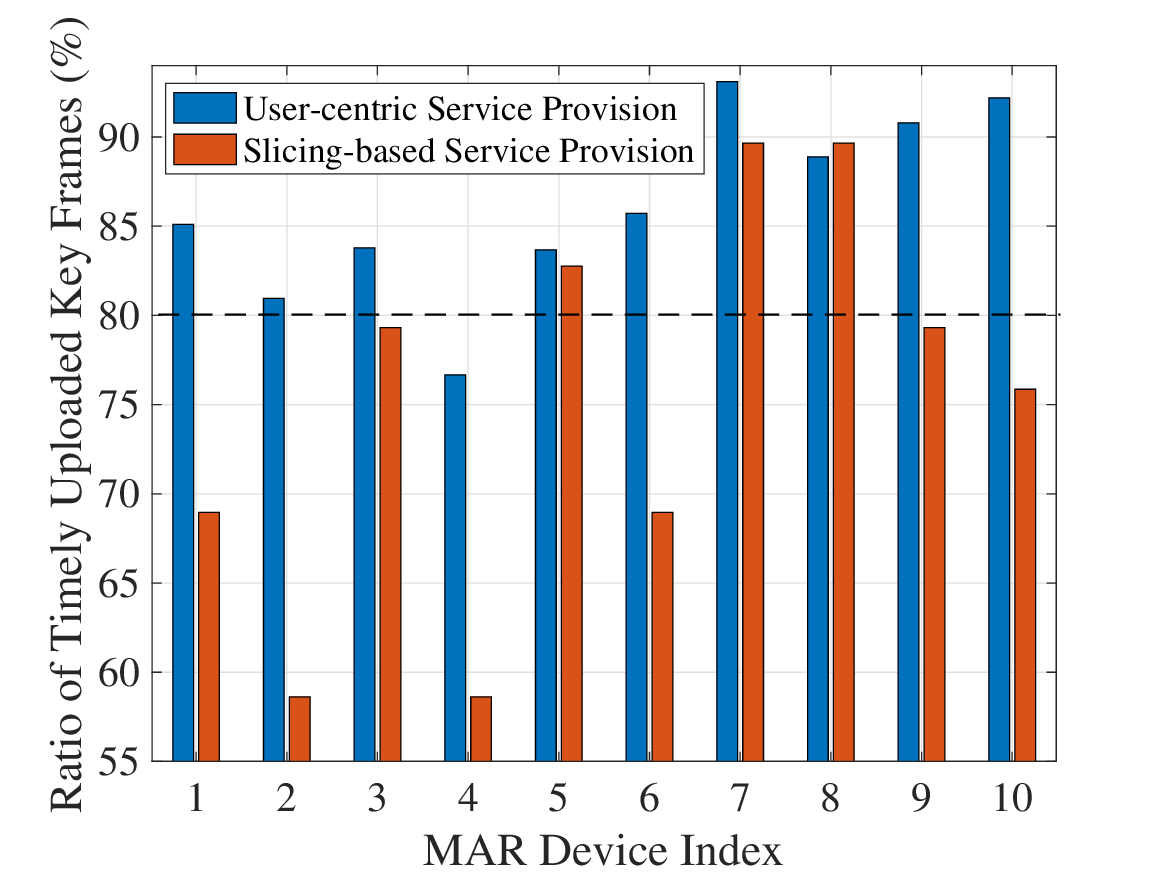}
            \caption{The ratio of timely uploaded key frames of different MAR devices when~$N = 10$.}
            \label{figc3}
        \end{subfigure}
        \caption{Performance comparison between user-centric and slicing-based communication service provision.}
        \label{figc13}
    \end{figure*}

In Fig.~\ref{figc13}, we compare the user-centric and the slicing-based communication service provision in terms of the TUKF and the RP. In Figs.~\ref{figc13}(a) and~\ref{figc13}(b), we plot the ratios as blue and red lines and illustrate the range between the minimum and maximum values using shaded blue and red areas, respectively, where each value is averaged over multiple MAR devices and 1,000 simulation runs. In Fig.~\ref{figc13}(a), we observe that, with the setting of~$\epsilon = 0.8$, the user-centric communication service provision significantly outperforms the slicing-based communication service provision in terms of the TUKF. This is because the proposed approach makes communication service provision decisions to meet the delay requirement of each MAR device based on user-specific service demands estimated by using MARLIN, rather than on averaged service demands across multiple MAR devices. In the case of $N = 10$ and~$\epsilon = 0.8$, the detailed performance differences across MAR devices are shown in Fig.~\ref{figc13}(c). In contrast to slicing-based communication service provision used for 5G, user-centric communication service provision ensures that nearly all MAR devices timely upload their key frames by adapting to user-specific service demands.

In Fig.~\ref{figc13}(b), the average amount of radio spectrum resource reserved for an MAR device in the user-centric communication service provision is less than that in slicing-based communication service provision. This demonstrates that user-centric communication service provision can adapt to user-specific service demand variances. Meanwhile, we observe that slicing-based communication service provision is less sensitive to the variances of individual service demands of MAR devices since the value of~$\Phi^{-1}(\epsilon) \sigma^{2}_{t}/|\mathcal{N}|$ decreases as the number of devices increases. Both observations align with Theorem~\ref{theorem2} and our theoretical analysis in Section~\ref{sec5}. 

\section{Conclusion and Future Work}\label{sec7}

In this paper, we have developed a novel user DT,~i.e., MARLIN, and a DT-based user-centric communication service provision approach to support edge-assisted device pose tracking in MAR. In MARLIN, the well-structured data model enables a flexible configuration of user-specific data attributes and personalized data update. The two designed DT operation functions jointly enable an adaptive switching between two data-driven modeling methods for accurately modeling non-stationary data traffic. Trace-driven simulation results have demonstrated the effectiveness and robustness of our user-centric communication service provision approach in coping with modeling inaccuracies. MARLIN provides an extensible prototype of user DT for capturing the impact of user-specific service demands and reveals the promising potential of fine-grained data management in facilitating user-centric communication service provision in the 6G. In the future, we plan to investigate a hybrid user-centric and slicing-based communication service provision scheme to enhance the quality of experience for MAR users while reducing the cost of data management incurred by user DTs.

\appendices

\section{Proof of Theorem~\ref{theorem1}}\label{appendix:the1} 
\begin{proof}
    
    According to Bayes' theorem, the true positive ratio (TPR) in time slot~$t$, i.e.,~$P(a_{f}=1|\hat{a}_{f} = 1)\,\, \forall f \in \mathcal{F}_{t}$, is given by: 
            \begin{subequations}\label{}
            \begin{align}
                p_{t}^\text{TPR} & =  \frac{P(\hat{a}_{f} | a_{f}=1) P(a_{f}=1)}{P(\hat{a}_{f}=1)}\\
                & = \frac{p_{t}\lambda_{t}}{p_{t}\lambda_{t} +(1-q_{t})(1-\lambda_{t})},
            \end{align}
        \end{subequations}
    where $p_{t} = P(\hat{a}_{f}=1 | a_{f}=1)$ and $q_{t} = P(\hat{a}_{f}=0 | a_{f}=0)$ represent the likelihoods, and $\lambda_{t} = P(a_{f}=1)$ denotes the prior probability in time slot~$t$. Similarly, we can calculate the true negative ratio (TNR) in time slot~$t$, i.e.,~$P(a_{f}=0|\hat{a}_{f} = 0)$, as follows:
           \begin{equation}\label{}
                 p_{t}^\text{TNR} = \frac{q_{t}(1-\lambda_{t})}{q_{t}(1-\lambda_{t}) + (1-p_{t})\lambda_{t}}.  
           \end{equation}

   Under the assumption that, given $\hat{a}_{f}, \forall f \in \mathcal{F}_{t}$, random variables~$a_{f}, \forall f \in \mathcal{F}_{t}$ are i.i.d., and $P(a_{f} | \hat{\mathbf{a}}_{t}) = P(a_{f} | \hat{a}_{f}), \forall f \in \mathcal{F}_{t}$, we obtain the following equation:
           \begin{equation}\label{}
                 P( \mathbf{a}_{t} | \hat{\mathbf{a}}_{t}) = \prod_{f \in \mathcal{F}_{t}}{P(a_{f}|\hat{a}_{f})}, \,\, \forall t \in \mathcal{T}.  
           \end{equation}
    Given $\hat{A}_{t}$, i.e., the number of camera frames predicted to be key frames, the number of true key frames among them follows a binomial distribution, i.e.,~$Y \sim B(\hat{A}_{t}, p_{t}^\text{TPR})$. Similarly, given the number of~$F_{t} - \hat{A}_{t}$ camera frames predicted to not be key frames, the number of true key frames among them follows a binomial distribution, i.e.,~$Z \sim B(F_{t} - \hat{A}_{t}, 1-p_{t}^\text{TNR})$. Thus,~$Y+Z$ represents the number of true key frames given the prediction result, and the probability~$P( \sum_{f \in \mathcal{F}_{t}}{a_{f} } \le k_{t} | \hat{\mathbf{a}}_{t} )$ is as follows: 
            \begin{equation}\label{eqa1}
                \begin{aligned} 
                    & P(Y+Z \leq k_{t} | \hat{\mathbf{a}}_{t}) = \sum_{k=0}^{k_{t}}{P(Y+Z = k | \hat{\mathbf{a}}_{t} )}\\
                    & = \sum_{k=0}^{N_{t}}{ \sum_{j = \max(0, k-(F_{t}-\hat{A}_{t}))}^{\min(\hat{A}_{t},k)}{P(Y=j)P(Z=k-j)}},\\
                \end{aligned} 
            \end{equation}
    where 
           \begin{equation}\label{}
                P(Y=j) = \binom{\hat{A}_{t}}{j} (p_{t}^\text{TPR})^{j}(1-p_{t}^\text{TPR})^{\hat{A}_{t}-j},   
           \end{equation}
    and
           \begin{equation}\label{}
                P(Z=k - j) = \binom{F_{t}-\hat{A}_{t}}{k-j} (1-p_{t}^\text{TNR})^{k-j}(p_{t}^\text{TNR})^{F_{t}-\hat{A}_{t}-k+j}.  
           \end{equation}
    Therefore,~$P( \sum_{f \in \mathcal{F}_{t}}{a_{f} } \le k_{t} | \hat{\mathbf{a}}_{t} )$ can be calculated based on parameters~$p_{t}$, $q_{t}$, and $\lambda_{t}$, as shown in~\eqref{eq30}.
\end{proof}

\section{Proof of Theorem~\ref{theorem2}}\label{appendix:the2} 
\begin{proof}
    The optimal solution to Problem~P3, i.e.,~$B^{\text{s}, *}_{t}$, is given by:
        \begin{equation}\label{}
          B^{\text{s}, *}_{t} = \frac{\alpha T^\text{r} |\mathcal{N}|}{\log (1 + \gamma_{t})}  k^{\text{s},*}_{t},\,\, \forall t \in \mathcal{T},
        \end{equation}
    where~$k^{\text{s},*}_{t}$ denotes the minimum value of $k^{\text{s}}_{t}$, given by:
        \begin{equation}\label{eqb2}
          k^{\text{s},*}_{t} =  \arg \min_{k^{\text{s}}_{t}} P \left( \frac{1}{|\mathcal{N}|} \sum_{n \in \mathcal{N}}{\tilde{k}_{n,t}} \leq k^{\text{s}}_{t} \right) \ge \epsilon. 
        \end{equation}
    We can derive \eqref{eqb2} as follows:
        \begin{equation}\label{}
          k^{\text{s},*}_{t} =  \arg \min_{k^{\text{s}}_{t}} g(k^{\text{s}}_{t}) \ge \epsilon,
        \end{equation}
    where~$g(k^{\text{s}})$ is the cumulative distribution function (CDF) of $\frac{1}{|\mathcal{N}|} \sum_{n \in \mathcal{N}}{\tilde{k}_{n,t}}$. Since CDF~$g(k^{\text{s}}_{t})$ is non-decreasing,
        \begin{equation}\label{}
          k^{\text{s},*}_{t} =  \lceil g^{-1}(\epsilon) \rceil, 
        \end{equation}
    where~$g^{-1}$ represents the inverse CDF, i.e., quantile function, and~$ \lceil \cdot \rceil$ is the ceiling function. 

  Under the assumption that random variables~$\tilde{k}_{n,t}, \forall n  \in \mathcal{N}$ are i.i.d with mean~$\bar{k}_{t}$ and variance~$\sigma^{2}_{t}$, their sum~$\frac{1}{|\mathcal{N}|} \sum_{n \in \mathcal{N}}{\tilde{k}_{n,t}}$ follows a normal distribution with mean~$\bar{k}_{t}$ and variance~$\sigma^{2}_{t}/|\mathcal{N}|$ as $|\mathcal{N}| \rightarrow \infty$, according to the central limit theorem. As a result, we can obtain the closed-form solution to~$k^{\text{s}}_{t}$ as follows:
        \begin{equation}\label{}
          k^{\text{s},*}_{t} =  \lceil \bar{k}_{t} + \frac{\sigma^{2}_{t}}{|\mathcal{N}|} \Phi^{-1}(\epsilon)\rceil, 
        \end{equation}
    where $\Phi^{-1}(\epsilon)$ represents the inverse CDF, i.e., a quantile function, of the standard normal distribution. 

    Generally, the value of SNR~$\gamma_{n,t}$ and the number of uploaded key frames~$\tilde{k}_{n,t}$ are independent. Given the mean SNR~$\bar{\gamma}_{t}$ over $|\mathcal{N}|$ MAR devices, the optimal solution to the slicing-based service provision in Problem~P3 can be derived based on~$k^{\text{s},*}_{t}$ in~\eqref{eq19}. 
\end{proof}






%
\bibliography{ref_AR2}

\begin{thebibliography}{10}
\providecommand{\url}[1]{#1}
\csname url@samestyle\endcsname
\providecommand{\newblock}{\relax}
\providecommand{\bibinfo}[2]{#2}
\providecommand{\BIBentrySTDinterwordspacing}{\spaceskip=0pt\relax}
\providecommand{\BIBentryALTinterwordstretchfactor}{4}
\providecommand{\BIBentryALTinterwordspacing}{\spaceskip=\fontdimen2\font plus
\BIBentryALTinterwordstretchfactor\fontdimen3\font minus
  \fontdimen4\font\relax}
\providecommand{\BIBforeignlanguage}[2]{{%
\expandafter\ifx\csname l@#1\endcsname\relax
\typeout{** WARNING: IEEEtran.bst: No hyphenation pattern has been}%
\typeout{** loaded for the language `#1'. Using the pattern for}%
\typeout{** the default language instead.}%
\else
\language=\csname l@#1\endcsname
\fi
#2}}
\providecommand{\BIBdecl}{\relax}
\BIBdecl

\bibitem{zhou2024user}
C.~Zhou, J.~Gao, Y.~Liu, S.~Hu, N.~Cheng, and X.~S. Shen, ``User-centric
  service provision for edge-assisted mobile {AR}: {A} digital twin-based
  approach,'' in \emph{IEEE/CIC ICCC)}.\hskip 1em plus 0.5em minus 0.4em\relax
  IEEE, Hangzhou, China, 2024.

\bibitem{zhou2024user_wcm}
C.~Zhou, S.~Hu, J.~Gao, X.~Huang, W.~Zhuang, and X.~Shen, ``User-centric
  immersive communications in {6G}: {A} data-oriented framework via digital
  twin,'' \emph{IEEE Wireless Commun.}, vol.~32, no.~3, pp. 122--129, 2025.

\bibitem{jin2023ebublio}
Y.~Jin, J.~Liu, F.~Wang, and S.~Cui, ``Ebublio: {Edge} assisted multi-user
  360-degree video streaming,'' \emph{IEEE IoT J.}, vol.~10, no.~17, pp.
  15\,408--15\,419, 2023.

\bibitem{huzaifa2021illixr}
M.~Huzaifa, R.~Desai, S.~Grayson, X.~Jiang, Y.~Jing, J.~Lee, F.~Lu, Y.~Pang,
  J.~Ravichandran, F.~Sinclair \emph{et~al.}, ``{ILLIXR}: {Enabling} end-to-end
  extended reality research,'' in \emph{Proc. IEEE IISWC}, Storrs, CT, USA,
  2021.

\bibitem{chen2023networked}
J.~Chen, K.~Ramakrishnan, A.~Dhakazl, and X.~Ran, ``Networked architectures for
  localization-based multi-user augmented reality,'' \emph{IEEE Commun. Mag.},
  vol.~61, no.~12, pp. 104--110, 2023.

\bibitem{wang2022leaf}
H.~Wang, B.~Kim, J.~Xie, and Z.~Han, ``{LEAF+AIO}: {Edge}-assisted energy-aware
  object detection for mobile augmented reality,'' \emph{IEEE Trans. Mobile
  Comput.}, vol.~22, no.~10, pp. 5933--5948, 2022.

\bibitem{wang2017joint}
F.~Wang, J.~Xu, X.~Wang, and S.~Cui, ``Joint offloading and computing
  optimization in wireless powered mobile-edge computing systems,'' \emph{IEEE
  Trans. Wirel. Commun.}, vol.~17, no.~3, pp. 1784--1797, 2017.

\bibitem{shen2021holistic}
X.~Shen, J.~Gao, W.~Wu, M.~Li, C.~Zhou, and W.~Zhuang, ``Holistic network
  virtualization and pervasive network intelligence for {6G},'' \emph{IEEE
  Commun. Surveys Tuts.}, vol.~24, no.~1, pp. 1--30, 2021.

\bibitem{ran2020multi}
X.~Ran, C.~Slocum, Y.-Z. Tsai, K.~Apicharttrisorn, M.~Gorlatova, and J.~Chen,
  ``Multi-user augmented reality with communication efficient and spatially
  consistent virtual objects,'' in \emph{Proc. ACM CoNEXT}, New York, NY, USA,
  2020.

\bibitem{navarro2020survey}
J.~Navarro-Ortiz, P.~Romero-Diaz, S.~Sendra, P.~Ameigeiras, J.~J. Ramos-Munoz,
  and J.~M. Lopez-Soler, ``A survey on {5G} usage scenarios and traffic
  models,'' \emph{IEEE Commun. Surveys Tuts.}, vol.~22, no.~2, pp. 905--929,
  2020.

\bibitem{campos2021orb}
C.~Campos, R.~Elvira, J.~J.~G. Rodr{\'\i}guez, J.~M. Montiel, and J.~D.
  Tard{\'o}s, ``{ORB-SLAM3}: {An} accurate open-source library for visual,
  visual--inertial, and multimap {SLAM},'' \emph{IEEE Trans. Robot.}, vol.~37,
  no.~6, pp. 1874--1890, 2021.

\bibitem{li2022risk}
X.~Li, R.~Xiao, M.~Pan, and N.~Zhao, ``Risk-averse investment strategy for
  {MEC} service provisioning: {A} data-driven distributionally robust
  solution,'' \emph{IEEE IoT J.}, vol.~9, no.~23, pp. 24\,148--24\,160, 2022.

\bibitem{li2020data}
L.~Li, D.~Shi, R.~Hou, X.~Li, J.~Wang, H.~Li, and M.~Pan, ``Data-driven
  optimization for cooperative edge service provisioning with demand
  uncertainty,'' \emph{IEEE IoT J.}, vol.~8, no.~6, pp. 4317--4328, 2020.

\bibitem{wu2023characterizing}
F.~Wu, F.~Lyu, J.~Ren, P.~Yang, K.~Qian, S.~Gao, and Y.~Zhang, ``Characterizing
  {Internet} card user portraits for efficient churn prediction model design,''
  \emph{IEEE Trans. Mobile Comput.}, vol.~23, no.~2, pp. 1735--1752, 2024.

\bibitem{shen2023toward}
X.~Shen, J.~Gao, M.~Li, C.~Zhou, S.~Hu, M.~He, and W.~Zhuang, ``Toward
  immersive communications in {6G},'' \emph{Front. Comput. Sci.}, vol.~4, 2023.

\bibitem{zhou2024digital}
C.~Zhou, J.~Gao, M.~Li, N.~Cheng, X.~Shen, and W.~Zhuang, ``Digital twin-based
  {3D} map management for edge-assisted device pose tracking in mobile {AR},''
  \emph{IEEE IoT J.}, vol.~11, no.~10, pp. 17\,812--17\,826, 2024.

\bibitem{sun2024knowledge}
R.~Sun, N.~Cheng, C.~Li, F.~Chen, and W.~Chen, ``Knowledge-driven deep learning
  paradigms for wireless network optimization in {6G},'' \emph{IEEE Netw.},
  2024, to be published, doi: 10.1109/MNET.2024.3352257.

\bibitem{linowes2017augmented}
J.~Linowes and K.~Babilinski, \emph{Augmented reality for developers: {Build}
  practical augmented reality applications with {Unity}, {ARCore}, {ARKit}, and
  {Vuforia}}.\hskip 1em plus 0.5em minus 0.4em\relax Packt Publishing Ltd,
  2017.

\bibitem{hu2023adaptive}
S.~Hu, M.~Li, J.~Gao, C.~Zhou, and X.~Shen, ``Adaptive device-edge
  collaboration on {DNN} inference in {AIoT}: {A} digital twin-assisted
  approach,'' \emph{IEEE IoT J.}, vol.~11, no.~7, pp. 12\,893--12\,908, 2023.

\bibitem{han2022comic}
B.~Han, P.~Pathak, S.~Chen, and L.-F.~C. Yu, ``{CoMIC}: {A} collaborative
  mobile immersive computing infrastructure for conducting multi-user {XR}
  research,'' \emph{IEEE Netw.}, pp. 1--9, 2022.

\bibitem{khosoussi2019reliable}
K.~Khosoussi, M.~Giamou, G.~S. Sukhatme, S.~Huang, G.~Dissanayake, and J.~P.
  How, ``Reliable graphs for {SLAM},'' \emph{The International Journal of
  Robotics Research}, vol.~38, no. 2-3, pp. 260--298, 2019.

\bibitem{davison2007monoslam}
A.~J. Davison, I.~D. Reid, N.~D. Molton, and O.~Stasse, ``{MonoSLAM}:
  {Real}-time single camera {SLAM},'' \emph{IEEE Trans. Pattern Anal. Mach.
  Intell.}, vol.~29, no.~6, pp. 1052--1067, 2007.

\bibitem{chou2021efficient}
C.~Chou and C.~Chou, ``Efficient and accurate tightly-coupled visual-{Lidar}
  {SLAM},'' \emph{IEEE Trans. Intell. Transp. Syst.}, vol.~23, no.~9, pp.
  14\,509--14\,523, 2021.

\bibitem{carlone2018attention}
L.~Carlone and S.~Karaman, ``Attention and anticipation in fast visual-inertial
  navigation,'' \emph{IEEE Trans. Robot.}, vol.~35, no.~1, pp. 1--20, 2018.

\bibitem{chen2018marvel}
K.~Chen, T.~Li, H.-S. Kim, D.~E. Culler, and R.~H. Katz, ``{Marvel}: {Enabling}
  mobile augmented reality with low energy and low latency,'' in \emph{Proc.
  ACM SenSys}, Shenzhen, China, 2018.

\bibitem{ben2022edge}
A.~J. Ben~Ali, M.~Kouroshli, S.~Semenova, Z.~S. Hashemifar, S.~Y. Ko, and
  K.~Dantu, ``{Edge}-{SLAM}: {Edge}-assisted visual simultaneous localization
  and mapping,'' \emph{ACM Trans. Embed. Comput. Syst.}, vol.~22, no.~1, pp.
  1--31, 2022.

\bibitem{chen2023adaptslam}
Y.~Chen, H.~Inaltekin, and M.~Gorlatova, ``{AdaptSLAM}: {Edge}-assisted
  adaptive {SLAM} with resource constraints via uncertainty minimization,'' in
  \emph{Proc. IEEE INFOCOM}, New York, NY, USA, 2023.

\bibitem{dhakal2022slam}
A.~Dhakal, X.~Ran, Y.~Wang, J.~Chen, and K.~Ramakrishnan, ``{SLAM}-share:
  {Visual} simultaneous localization and mapping for real-time multi-user
  augmented reality,'' in \emph{Proc. ACM CoNEXT}, Rome, Italy, 2022.

\bibitem{ren2020edge}
P.~Ren, X.~Qiao, Y.~Huang, L.~Liu, C.~Pu, S.~Dustdar, and J.~Chen, ``{Edge}
  {AR} {X5}: {An} edge-assisted multi-user collaborative framework for mobile
  web augmented reality in {5G} and beyond,'' \emph{IEEE Trans. Cloud Comput.},
  vol.~10, no.~4, pp. 2521--2537, 2020.

\bibitem{lin2018towards}
S.-C. Lin, P.~Wang, I.~F. Akyildiz, and M.~Luo, ``Towards optimal network
  planning for software-defined networks,'' \emph{IEEE Trans. Mobile Comput.},
  vol.~17, no.~12, pp. 2953--2967, 2018.

\bibitem{wu2022demand}
S.-H. Wu, C.-H. Ko, and H.-L. Chao, ``On-demand coordinated spectrum and
  resource provisioning under an open {C-RAN} architecture for dense small cell
  networks,'' \emph{IEEE Trans. Mobile Comput.}, vol.~23, no.~1, pp. 673--688,
  2024.

\bibitem{zhang2020latency}
L.~Zhang and N.~Ansari, ``Latency-aware {IoT} service provisioning in
  {UAV}-aided mobile-edge computing networks,'' \emph{IEEE IoT J.}, vol.~7,
  no.~10, pp. 10\,573--10\,580, 2020.

\bibitem{qu2021service}
Y.~Qu, H.~Dai, H.~Wang, C.~Dong, F.~Wu, S.~Guo, and Q.~Wu, ``Service
  provisioning for {UAV}-enabled mobile edge computing,'' \emph{IEEE J. Sel.
  Areas Commun.}, vol.~39, no.~11, pp. 3287--3305, 2021.

\bibitem{zhang2021aoi}
X.~Zhang, J.~Wang, and H.~V. Poor, ``{AoI}-driven statistical delay and
  error-rate bounded {QoS} provisioning for {mURLLC} over {UAV}-multimedia {6G}
  mobile networks using {FBC},'' \emph{IEEE J. Sel. Areas Commun.}, vol.~39,
  no.~11, pp. 3425--3443, 2021.

\bibitem{qu2021resilient}
Y.~Qu, D.~Lu, H.~Dai, H.~Tan, S.~Tang, F.~Wu, and C.~Dong, ``Resilient service
  provisioning for edge computing,'' \emph{IEEE IoT J.}, vol.~10, no.~3, pp.
  2255--2271, 2021.

\bibitem{ye2018dynamic}
Q.~Ye, W.~Zhuang, S.~Zhang, A.-L. Jin, X.~Shen, and X.~Li, ``Dynamic radio
  resource slicing for a two-tier heterogeneous wireless network,'' \emph{IEEE
  Trans. Veh. Technol.}, vol.~67, no.~10, pp. 9896--9910, 2018.

\bibitem{jia2021vnf}
Z.~Jia, M.~Sheng, J.~Li, D.~Zhou, and Z.~Han, ``{VNF}-based service provision
  in software defined {LEO} satellite networks,'' \emph{IEEE Trans. Wirel.
  Commun.}, vol.~20, no.~9, pp. 6139--6153, 2021.

\bibitem{qiu2022online}
Y.~Qiu, J.~Liang, V.~C. Leung, X.~Wu, and X.~Deng, ``Online
  reliability-enhanced virtual network services provisioning in fault-prone
  mobile edge cloud,'' \emph{IEEE Trans. Wirel. Commun.}, vol.~21, no.~9, pp.
  7299--7313, 2022.

\bibitem{wang2020sfc}
G.~Wang, S.~Zhou, S.~Zhang, Z.~Niu, and X.~Shen, ``{SFC}-based service
  provisioning for reconfigurable space-air-ground integrated networks,''
  \emph{IEEE J. Sel. Areas Commun.}, vol.~38, no.~7, pp. 1478--1489, 2020.

\bibitem{belgiovine2021deep}
M.~Belgiovine, K.~Sankhe, C.~Bocanegra, D.~Roy, and K.~R. Chowdhury, ``Deep
  learning at the edge for channel estimation in beyond-{5G} massive {MIMO},''
  \emph{IEEE Wireless Commun.}, vol.~28, no.~2, pp. 19--25, 2021.

\bibitem{atawia2016joint}
R.~Atawia, H.~Abou-Zeid, H.~S. Hassanein, and A.~Noureldin, ``Joint
  chance-constrained predictive resource allocation for energy-efficient video
  streaming,'' \emph{IEEE J. Sel. Areas Commun.}, vol.~34, no.~5, pp.
  1389--1404, 2016.

\bibitem{kipf2016semi}
T.~N. Kipf and M.~Welling, ``Semi-supervised classification with graph
  convolutional networks,'' \emph{arXiv:1609.02907}, 2017, [Online]. Available:
  https://doi.org/10.48550/arXiv.1609.02907.

\bibitem{10771784}
Y.~Dai, L.~Lyu, N.~Cheng, M.~Sheng, J.~Liu, X.~Wang, S.~Cui, L.~Cai, and
  X.~Shen, ``A survey of graph-based resource management in wireless networks -
  {Part} {II}: {Learning} approaches,'' \emph{IEEE Trans. Cogn. Commun. Netw.},
  pp. 1--23, 2024.

\bibitem{zhang2019graph}
S.~Zhang, H.~Tong, J.~Xu, and R.~Maciejewski, ``Graph convolutional networks:
  {A} comprehensive review,'' \emph{Computational Social Networks}, vol.~6,
  no.~1, pp. 1--23, 2019.

\bibitem{hochreiter1997long}
S.~Hochreiter and J.~Schmidhuber, ``Long short-term memory,'' \emph{Neural
  computation}, vol.~9, no.~8, pp. 1735--1780, 1997.

\bibitem{sak2014long}
H.~Sak, A.~W. Senior, and F.~Beaufays, ``Long short-term memory based recurrent
  neural network architectures for large vocabulary speech recognition,''
  \emph{arXiv:1402.1128}, 2014, [Online]. Available:
  https://doi.org/10.48550/arXiv.1402.1128.

\bibitem{li2018interiornet}
W.~Li, S.~Saeedi, J.~McCormac, R.~Clark, D.~Tzoumanikas, Q.~Ye, Y.~Huang,
  R.~Tang, and S.~Leutenegger, ``{InteriorNet}: {Mega}-scale multi-sensor
  photo-realistic indoor scenes dataset,'' in \emph{British Machine Vision
  Conference}, 2018, Newcastle, UK.

\bibitem{ma2023nomore}
X.~Ma, Q.~Zeng, H.~Chi, and L.~Luo, ``No more companion {Apps} hacking but one
  dongle: {Hub-based} blackbox fuzzing of {loT} firmware,'' in \emph{Proc.~ACM
  MobiSys}, Helsinki, Finland, 2023.

\bibitem{ma2025matterguard}
X.~Ma, J.~Chen, L.~Luo, and Q.~Zeng, ``{LLM}-assisted {IoT} testing: {Finding}
  conformance bugs in {Matter} {SDKs},'' in \emph{Proc. ACM MobiCom}, Hong
  Kong, China, 2025.

\bibitem{wang2020long}
Z.~Wang, X.~Su, and Z.~Ding, ``Long-term traffic prediction based on {LSTM}
  encoder-decoder architecture,'' \emph{IEEE Trans. Intell. Transp. Syst.},
  vol.~22, no.~10, pp. 6561--6571, 2021.

\bibitem{oreshkinn}
B.~N. Oreshkin, D.~Carpov, N.~Chapados, and Y.~Bengio, ``{N-BEATS}: Neural
  basis expansion analysis for interpretable time series forecasting,'' in
  \emph{ICLR}, 2020, Addis Ababa, Ethiopia.

\bibitem{zhou2021informer}
H.~Zhou, S.~Zhang, J.~Peng, S.~Zhang, J.~Li, H.~Xiong, and W.~Zhang,
  ``Informer: Beyond efficient transformer for long sequence time-series
  forecasting,'' in \emph{Proc. AAAI}, 2021, Washington, DC, USA.

\end{thebibliography}

\bibliographystyle{IEEEtran}

\end{document}